\documentclass{article}

\usepackage{hyperref}
\usepackage{verbatim}
\usepackage{amsmath, amsthm, amssymb, graphics, mathtools, graphicx, xcolor}
\usepackage{enumitem}

\newtheorem{theorem}{Theorem}[section]
\newtheorem{lemma}[theorem]{Lemma}
\newtheorem{construction}[theorem]{Construction}
\newtheorem{corollary}[theorem]{Corollary}
\newtheorem{definition}[theorem]{Definition}
\newtheorem{conjecture}[theorem]{Conjecture}

\newtheorem{observation}[theorem]{Observation}
\newtheorem{claim}{Claim}

\DeclarePairedDelimiter\floor{\lfloor}{\rfloor}

\newcommand{\bpc}{\noindent {\em Proof of Claim~\theclaim. }}
\newcommand{\epc}{This proves Claim~\theclaim.}
\newcommand{\cpl}{This completes the proof.}
\newcommand{\sm}{\setminus}
\newcommand{\tri}{\Delta}
\newcommand{\seq}{\subseteq}
\newcommand{\myBox}{\textsc{Box}}
\newcommand{\Dom}{\textsc{Dom}}
\newcommand{\Scope}{\textsc{Scp}}

\newcommand{\bigO}{\mathcal{O}}

\DeclareMathOperator{\tw}{tw}
\DeclareMathOperator{\cw}{cw}
\DeclareMathOperator{\rw}{rw}
\DeclareMathOperator{\pw}{pw}
\DeclareMathOperator{\rank}{rank}

\makeatletter
\newtheorem*{rep@theorem}{\rep@title}
\newcommand{\newreptheorem}[2]{%
\newenvironment{rep#1}[1]{%
 \def\rep@title{#2 \ref{##1}}%
 \begin{rep@theorem}}%
 {\end{rep@theorem}}}
\makeatother

\newreptheorem{theorem}{Theorem}
\newreptheorem{corollary}{Corollary}
\newreptheorem{lemma}{Lemma}

\begin{document}

\title{(Theta, triangle)-free and (even hole, $K_4$)-free graphs. Part
  1: Layered wheels}

\author{Ni Luh Dewi Sintiari\thanks{Univ Lyon, EnsL, UCBL, CNRS, LIP,
    F-69342, LYON Cedex 07, France. \newline The authors are partially
    supported by the LABEX MILYON (ANR-10-LABX-0070) of Universit\'e
    de Lyon within the program ‘‘Investissements d'Avenir’’
    (ANR-11-IDEX-0007) operated by the French National Research Agency
    (ANR), and by Agence Nationale de la Recherche (France) under 
    research grant ANR DIGRAPHS ANR-19-CE48-0013-01.}~, 
    Nicolas Trotignon\footnotemark[1]}

\maketitle

\begin{abstract}
  We present a construction called layered wheel. Layered wheels are
  graphs of arbitrarily large treewidth and girth. They might be an
  outcome for a possible theorem characterizing graphs with large
  treewidth in terms of their induced subgraphs (while such a
  characterization is well-understood in terms of minors).  They also
  provide examples of graphs of large treewidth and large rankwidth in
  well-studied classes, such as (theta, triangle)-free graphs and
  even-hole-free graphs with no $K_4$ (where a hole is a chordless
  cycle of length at least four, a theta is a graph made of three
  internally vertex disjoint paths of length at least two linking two
  vertices, and $K_4$ is the complete graph on four vertices).
\end{abstract}

\section{Introduction}
\label{sec:introduction}

In this article, all graphs are finite, simple, and undirected. The
vertex set of a graph $G$ is denoted by $V(G)$ and the edge set by
$E(G)$.  A graph $H$ is an {\em induced subgraph} of a graph $G$ if
some graph isomorphic to $H$ can be obtained from $G$ by deleting
vertices. A graph $H$ is a {\em minor} of a graph $G$ if some graph
isomorphic to $H$ can be obtained from $G$ by deleting vertices,
deleting edges, and contracting edges.

When we say that {\em $G$ contains $H$} without specifying {\em as a
  minor} or {\em as an induced subgraph}, we mean that $H$ is an
induced subgraph of $G$.  A graph is {\em $H$-free} if it does not
contain $H$ (so, as an induced subgraph). For a family of graphs
${\mathcal H}$, $G$ is {\em ${\mathcal H}$-free} if for every
$H\in {\mathcal H}$, $G$ is $H$-free.  A class of graphs is {\em
  hereditary} if it is ${\mathcal H}$-free for some ${\mathcal H}$ or,
equivalently, if it is closed under taking induced subgraphs.  A {\em
  hole} in a graph is a chordless cycle of length at least four. It is
\emph{odd} or \emph{even} according to its length (that is its number
of edges). We denote by $K_\ell$ the complete graph on $\ell$
vertices. 

The present work is originally motivated by a question asked by
Cameron et al. in~\cite{CameronCH18}: is the
treewidth (or cliquewidth) of an even-hole-free graph bounded by a
function of its clique number? In this first part, we describe a
construction called \emph{layered wheel} showing that the answer is
no.  In the second part, we will show that under additional
restrictions, the treewidth is bounded. We postpone the formal
definition of a layered wheel to Section~\ref{sec:construction}
although we use the term several times until then. There are three
main motivations:

\begin{itemize}
\item When considering the induced subgraph relation (instead of the
  minor relation), is there a theorem similar to the celebrated
  grid-minor theorem of Robertson and Seymour?
\item A better understanding of the classes defined by excluding the
  so-called Truemper configurations, that play an important role in
  hereditary classes of graphs.
\item The structure of even-hole-free graphs.
\end{itemize}

We now give details on each of the three items.

\subsection*{The grid-minor theorem}

The {\em treewidth} of a graph is an integer measuring how far is the
graph from being a tree (far here means the difficulty of decomposing the graph
in a kind of tree-structure).  We give a formal definition of treewidth in
Section~\ref{sec:terminology}. 

The $(k \times k)$-{\em grid} is the graph on
$\{(i, j): 1\leq i, j \leq k\}$ where two distinct ordered pairs
$(i, j)$ and $(i', j')$ are adjacent whenever exactly one of the
following holds: $|i-i'| = 1$ and $j=j'$, or $i =i'$ and $|j-j'| =1$
(see Figure~\ref{fig:grid}). Robertson and
Seymour~\cite{RobertsonS86} proved that there exists
a function $f$ such that every graph with treewidth at least~$f(k)$
contains a $(k \times k)$-grid as a minor (see~\cite{chuzhoy:16} for
the best function known so far).  This is called the {\em grid-minor
  theorem}.  The $(k \times k)$-{\em wall} is the graph obtained from
the $(k \times k)$-grid by deleting all edges with form
$(2i+1, 2j)-(2i+1, 2j+1)$ and $(2i, 2j+1)-(2i, 2j+2)$.

\begin{figure}
  \begin{center}
    \includegraphics[height=2.5cm]{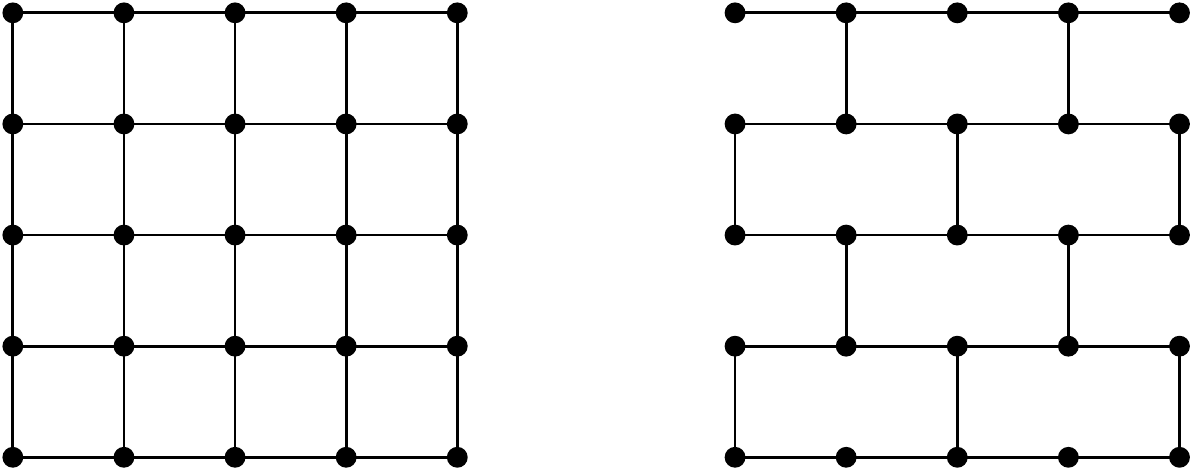}
  \end{center}  
  \caption{A grid and a wall\label{fig:grid}}
\end{figure}

{\em Subdividing $k$ times} an edge $e=uv$ of a graph, where
$k\geq 1$, means deleting $e$ and adding a path $u w_1\dots w_k v$.
The {\em $k$-subdivision} of a graph $G$ is the graph obtained from
$G$ by subdividing $k$-times all its edges (simultaneously).  Note
that replacing ``grid'' by a more specific graph in the grid-minor
theorem, such as $k$-subdivision of a $(k \times k)$-grid,
$(k \times k)$-wall, or $k$-subdivision of a $(k \times k)$-wall 
provides statements that are formally weaker (at the expense of a
larger function), because a large grid contains a large subdivision
of a grid, a large wall, and a large subdivision of a wall.  However,
these trivial corollaries are in some sense stronger, because walls,
subdivisions of walls, and subdivision of grids are graphs of large
treewidth that are more sparse than grids. So they somehow certify a
large treewidth with less information. Since one can always subdivide
more, there is no ``ultimate'' theorem in this direction.

\begin{figure}
  \begin{center}
    \includegraphics[height=2.5cm]{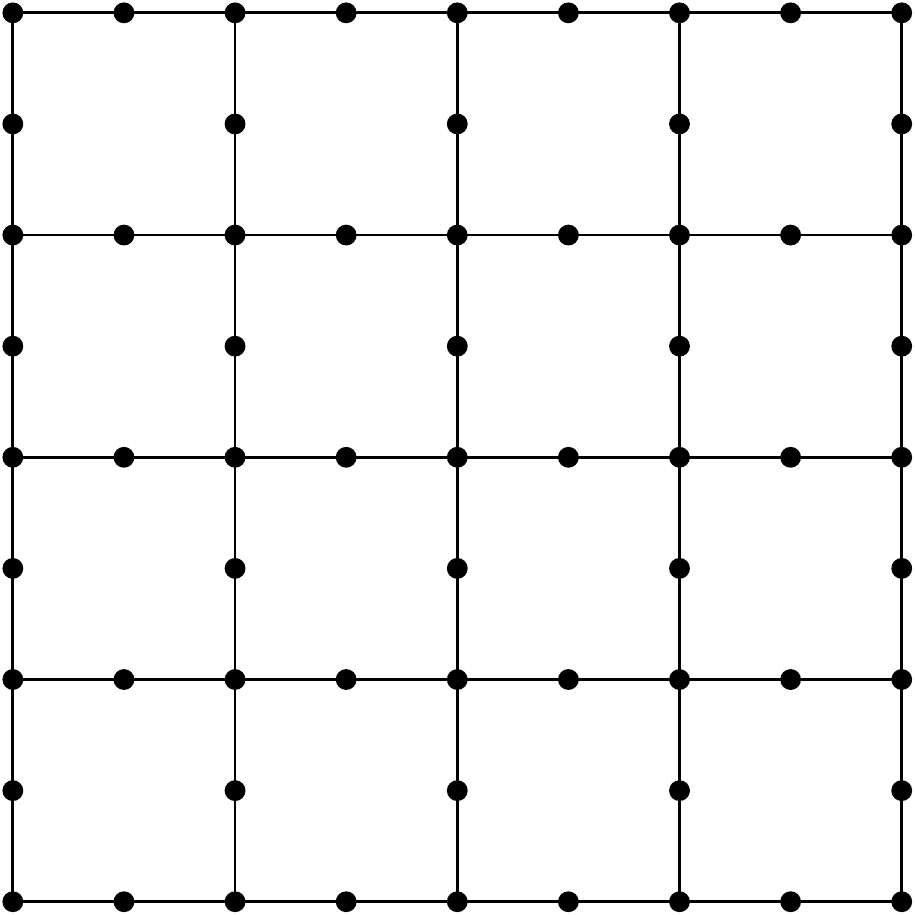}
    \hspace{.2em}
    \includegraphics[height=2.5cm]{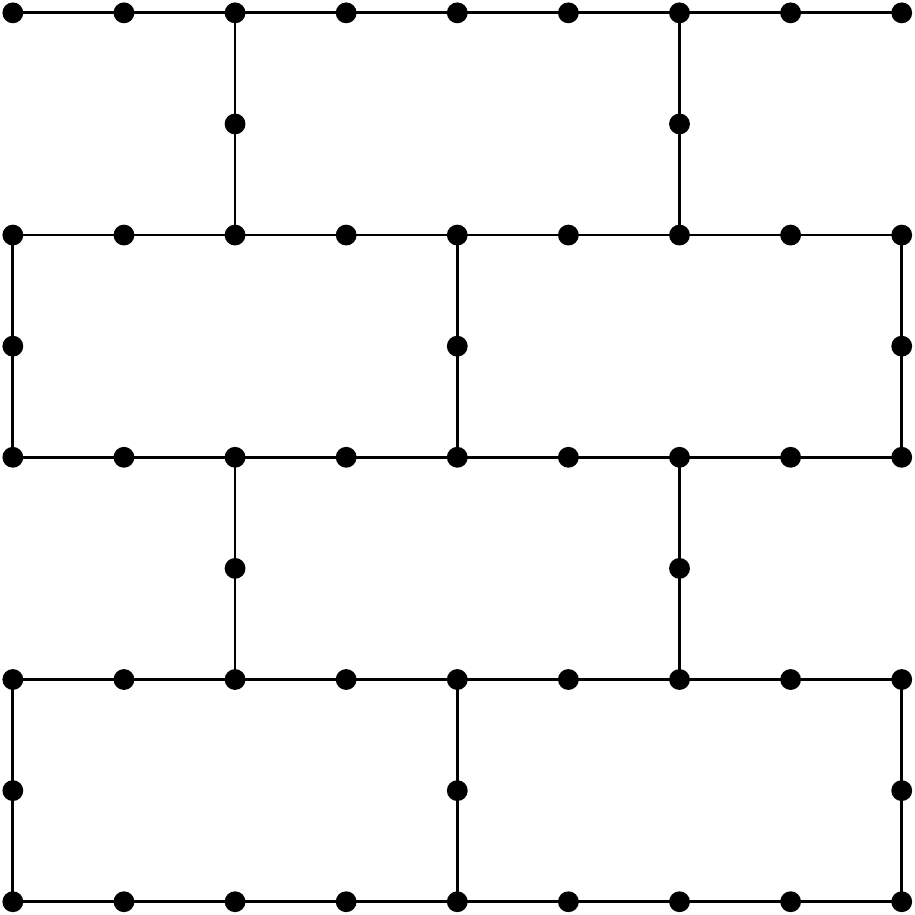}
    \hspace{.2em}
    \includegraphics[height=2.5cm]{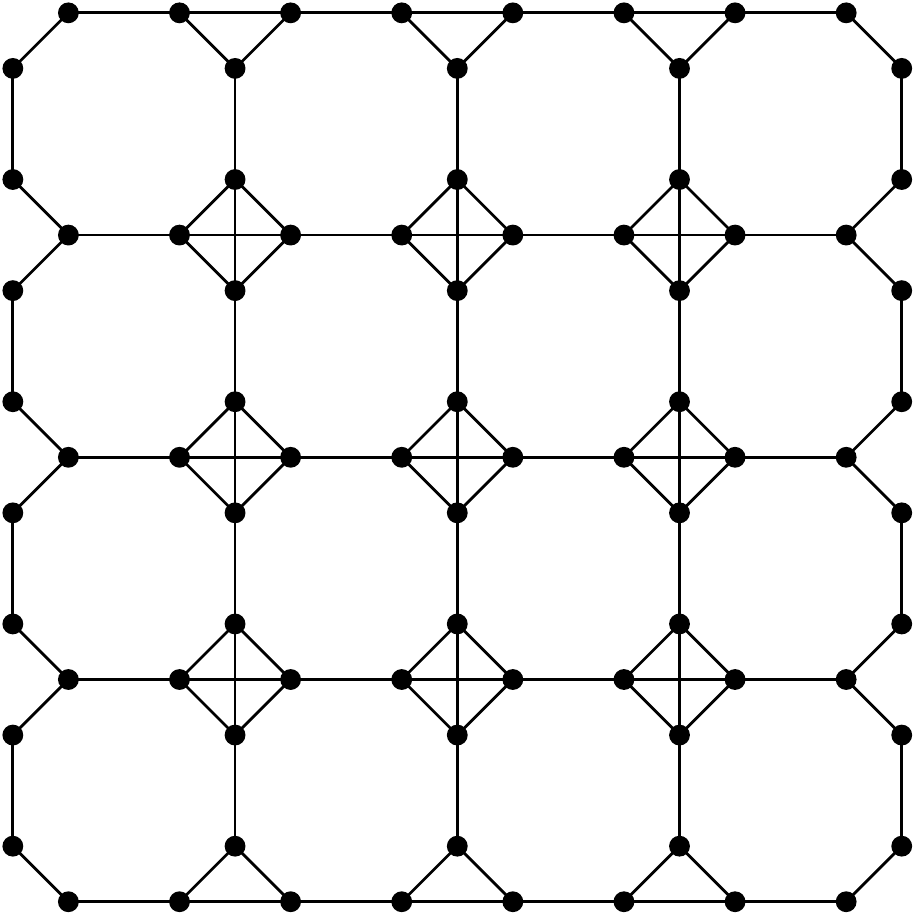}
    \hspace{.2em}
    \includegraphics[height=2.5cm]{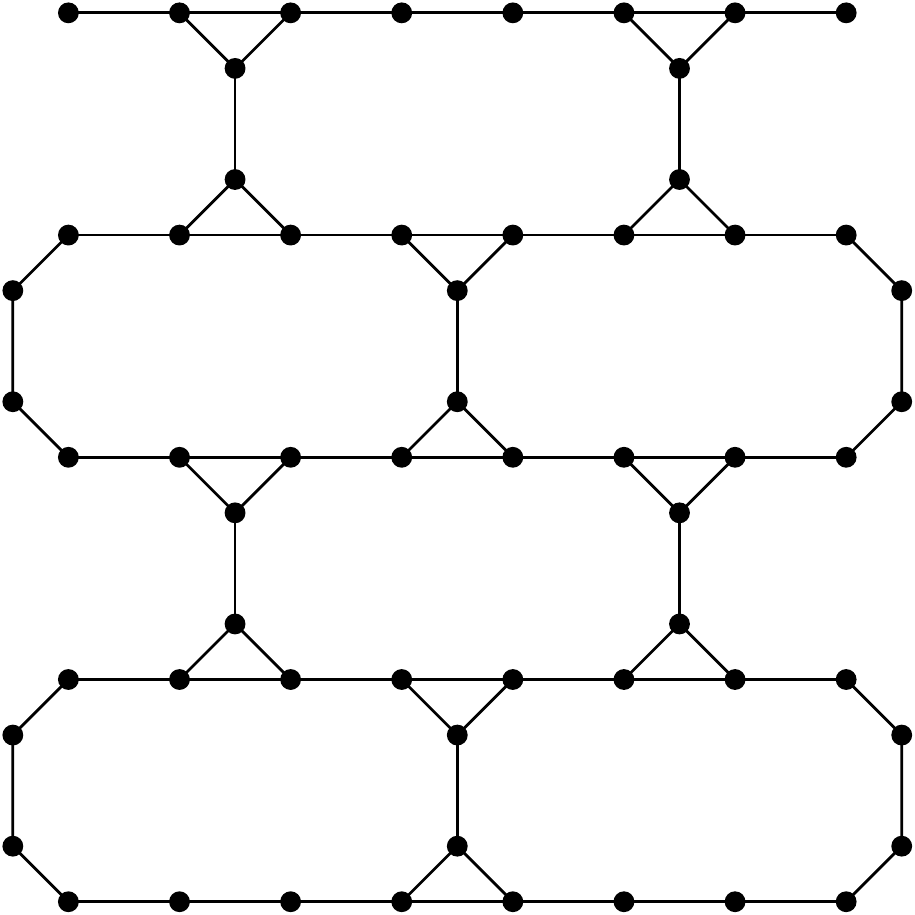}
  \end{center}
  \caption{A subdivision of a grid, of a wall, and the line graphs of the
    former\label{fig:grid-wall-subdiv-linegraph}}
\end{figure}

It would be useful to have a similar theorem with ``induced subgraph''
instead of ``minor''.  Simply replacing ``minor'' with ``induced
subgraph'' in the statement is trivially false, and here is a list of
known counter-examples: $K_k$, $K_{k,k}$, subdivisions of walls, line
graphs of subdivisions of walls (see Figure~\ref{fig:grid-wall-subdiv-linegraph}), where
$K_k$ denotes the complete graph on $k$ vertices, 
$K_{k, k}$ denotes the complete bipartite graph with each side of size $k$, and where the
{\em line graph} of a graph $R$ is the graph $G$ on $E(R)$ where two
vertices in $G$ are adjacent whenever they are adjacent edges of $R$.

One of our results is that the simple list above is not complete.  In
section~\ref{sec:layeredWheels}, we present a construction that we
call {\em layered wheel}.  Layered wheels have large treewidth and
large girth (the {\em girth} of a graph is the length of its
shortest cycle).  Large girth implies that they contain no $K_k$, no
$K_{k,k}$, and no line graphs of subdivisions of walls. Moreover,
layered wheels contain no subdivisions of $(3, 5)$-grids (this is 
explained after Lemma~\ref{lem:ttf-exist}).

We leave an open question asked by Zden\v ek Dvo\v r\'ak (personal
communication): is it true that for some function $f$ every graph with
treewidth at least~$f(k)$ contains either $K_k$, $K_{k, k}$, a
subdivision of the $(k \times k)$-wall, the line graph of some
subdivision of the $(k \times k)$-wall, or some variant of the layered
wheel with at least~$k$ layers?  In the next paragraphs, we give variants
of Dvo\v r\'ak's question.

\subsection*{Truemper configurations}

A {\em prism} is a graph made of three vertex-disjoint chordless paths
$P_1 = a_1 \dots b_1$, $P_2 = a_2 \dots b_2$, $P_3 = a_3 \dots b_3$ of
length at least 1, such that $a_1a_2a_3$ and $b_1b_2b_3$ are triangles
and no edges exist between the paths except those of the two
triangles.  Such a prism is also referred to as a
$3PC(a_1a_2a_3,b_1b_2b_3)$ or a $3PC(\Delta ,\Delta )$ (3PC stands for
{\em 3-path-configuration}).

A {\em pyramid} is a graph made of three chordless paths
$P_1 = a \dots b_1$, $P_2 = a \dots b_2$, $P_3 = a \dots b_3$ of
length at least one, two of which have length at least two, vertex-disjoint
except at $a$, and such that $b_1b_2b_3$ is a triangle and no edges
exist between the paths except those of the triangle and the three
edges incident to $a$.  Such a pyramid is also referred to as a
$3PC(b_1b_2b_3,a)$ or a $3PC(\Delta ,\cdot)$.

A {\em theta} is a graph made of three internally vertex-disjoint
chordless paths $P_1 = a \dots b$, $P_2 = a \dots b$,
$P_3 = a \dots b$ of length at least two and such that no edges exist
between the paths except the three edges incident to $a$ and the three
edges incident to $b$.  Such a theta is also referred to as a
$3PC(a, b)$ or a $3PC(\cdot ,\cdot)$.

\begin{figure}
  \begin{center}
    \includegraphics[height=2cm]{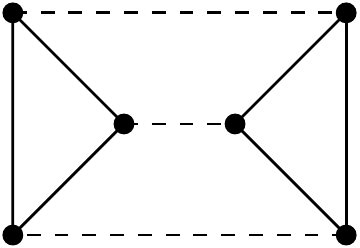}
    \hspace{.2em}
    \includegraphics[height=2cm]{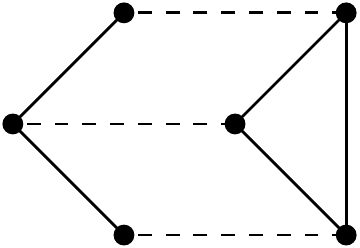}
    \hspace{.2em}
    \includegraphics[height=2cm]{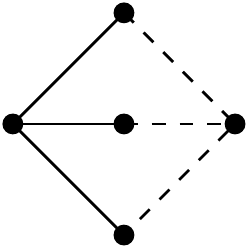}
    \hspace{.2em}
    \includegraphics[height=2cm]{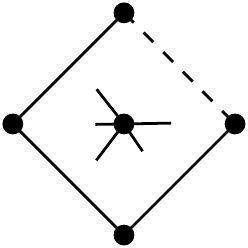}
  \end{center}
  \caption{Prism, pyramid, theta, and wheel (dashed lines represent
    paths)\label{fig:truemper-config}} 
\end{figure}

Observe that the lengths of the paths in the three definitions above
are designed so that the union of any two of the paths induces a hole.  A
{\em wheel} $W = (H, c)$ is a graph formed by a hole $H$ (called the
{\em rim}) together with a vertex $c$ (called the {\em center}) that
has at least three neighbors in the hole.

A {\em 3-path-configuration} is a graph isomorphic to a prism, a
pyramid, or a theta.  A~{\em Truemper configuration} is a graph
isomorphic to a prism, a pyramid, a theta, or a wheel.  They appear in
a theorem of Truemper~\cite{truemper} that characterizes graphs whose
edges can be labeled so that all chordless cycles have prescribed
parities (3-path-configurations seem to have first appeared in a
paper Watkins and Mesner~\cite{watkinsMesner:cycle}).

Truemper configurations play an important role in the analysis of
several important hereditary graph classes, as explained in a survey
of Vu\v skovi\'c~\cite{vuskovic:truemper}.  Let us simply mention here
that many decomposition theorems for classes of graphs are proved by
studying how some Truemper configurations contained in the graph
attaches to the rest of the graph, and often, the study relies on the
fact that some other Truemper configurations are excluded from the
class.  The most famous example is perhaps the class of {\em perfect
  graphs}.  In these graphs, pyramids are excluded, and how a prism
contained in a perfect graphs attaches to the rest of the graph is
important in the decomposition theorem for perfect graphs, whose
corollary is the celebrated {\em Strong Perfect Graph Theorem} due to
Chudnovsky, Robertson, Seymour, and
Thomas~\cite{chudnovsky.r.s.t:spgt}.  See also~\cite{nicolas:perfect}
for a survey on perfect graphs, where a section is specifically
devoted to Truemper configurations.  Many other examples exist, see
\cite{diotRaTrVu:15} for a long list of them.

Some researchers started to study systematically classes defined by
excluding some Truemper configurations~\cite{diotRaTrVu:15}.
We believe that among many classes that can be defined in that way,
the class of theta-free graphs is one of the most interesting classes.  This
is because it generalizes claw-free graphs (since a theta contains a
claw), and so it is natural to ask whether it shares the most interesting 
features of claw-free graphs: a structural description (see
\cite{chudnovsky.seymour:claw4}), a polynomial time algorithm for the
maximum stable set (see~\cite{faenzaOrioloStauffer:clawFree}), an
approximation algorithms for the chromatic number~(see
\cite{king:these}), a polynomial time algorithm for the induced
linkage problem (see~\cite{fialaKLP:12}), and a polynomial
$\chi$-bounding function (see~\cite{gyarfas:perfect}).

In the attempt of finding a structural description of theta-free
graphs, a seemingly easy case is when triangles are also excluded.
Because then, every vertex of degree at least~3 is the center of a
claw (therefore a possible start for a theta), so that excluding
theta and triangle should enforce some structure.  Supporting this
idea, Radovanovi\'c and Vu\v skovi\'c~\cite{radovanovicV:theta} proved
that every (theta, triangle)-free graph is 3-colorable.

Hence, we believed when starting this work that (theta, triangle)-free
graphs have bounded treewidth. But this turned out to be false:
layered wheels are (theta, triangle)-free graphs of arbitrarily large
treewidth.

However, on the positive side, we note that layered wheels need many
vertices to increase the treewidth.  More specifically, a layered
wheel $G$ is made of $l+1$ layers, where $l$ is an integer.  Each layer
is a path and $|V(G)|\geq 2^l$
(see Lemma~\ref{lem:neighborhood-ttf-layered-wheel}), $l \leq \tw(G) \leq 2l$
(see Theorems~\ref{th:largeTW} and~\ref{th:bounded-pw}).  So, the
treewidth of a layered wheel is ``small'' in the sense that it is
logarithmic in the size of its vertex set. We wonder whether such a
behavior is general in the sense of the following conjecture.

\begin{conjecture}
  \label{conj:log-treewidth-ttf}
  For some constant $c$, if $G$ is a (theta, triangle)-free graph,
  then the treewidth of $G$ is at most $c\log |V(G)|$.
\end{conjecture}

This conjecture reflects our belief that constructions similar to the
layered wheel must have an exponential number of vertices (exponential
in the treewidth). It suggests the following variant of Dvo\v r\'ak's
question: is it true that for some constant $c>1$ and some function
$f$, every graph with treewidth at least~$f(k)$ contains either
$K_k$, $K_{k, k}$, a subdivision of the $(k \times k)$-wall, the line
graph of some subdivision of the $(k \times k)$-wall, or has at
least~$c^{f(k)}$ vertices?

Kristina Vu\v skovi\'c observed that $K_{k, k}$ is a (prism,
pyramid, wheel)-free graph, or equivalently an \emph{only-theta graph}
(because the theta is the only Truemper configuration contained in
it). Moreover, walls are only-theta graphs, line graphs of subdivisions of
walls are only-prism graphs, and triangle-free layered wheels are only-wheel
graphs. Observe that complete graphs contain no Truemper
configuration, so they are simultaneously only-prism, only-wheel,
and only-theta. One may wonder whether a graph with large treewidth
should contain an induced subgraph of large treewidth with a
restricted list of induced subgraphs isomorphic to Truemper
configurations.

\subsection*{Even-hole-free graphs}

Our last motivation for this work is a better understanding of
even-hole-free graphs.  These are related to Truemper configurations
because thetas and prisms obviously contain even holes (to see this,
consider two paths of the same parity among the three paths that form the
configuration). Also, call {\em even wheel} a wheel $W= (H, c)$ where
$c$ has an even number of neighbors in $H$.  It is easy to check that
every even wheel contains an even hole.

Even-hole-free graphs were originally studied to experiment techniques
that would help to settle problems on perfect graphs. This has
succeeded, in the sense that the decomposition theorem for
even-hole-free graphs (see~\cite{vuskovic:evensurvey}) is in some
respect similar to the one that was later on discovered for perfect
graphs (see~\cite{chudnovsky.r.s.t:spgt}).  However, classical
problems such as graph coloring or maximum stable set, are polynomial
time solvable for perfect graphs, while they are still open for
even-hole-free graphs.  This is a bit strange because the
decomposition theorem for even-hole-free graphs is in many respect
simpler than the one for perfect graphs.  Moreover, it is easy to
provide perfect graphs of arbitrarily large treewidth (or even
rankwidth), such as bipartite graphs, or their line graphs.  On the other hand,
for even-hole-free graphs, apart from complete graphs, it is not so easy.
Some constructions are known, see~\cite{adlerLMRTV:rwehf}.

But so far, every construction of even-hole-free graphs of arbitrarily
large treewidth (or rankwidth) contains large cliques.  Moreover, it
is proved in~\cite{CameronSHV18} that (even hole,
triangle)-free graphs have bounded treewidth.  This is based on a
structural description of the class
from~\cite{ConfortiCKV00}.  Hence, Cameron et
al.~\cite{CameronCH18} asked whether (even hole,
$K_4$)-free graphs have bounded treewidth.  We prove in this article
that it is not the case, by a variant of the layered wheel
construction (see Theorem~\ref{th:layeredWheel-EHF}). 
As for (theta, triangle)-free, we need a large number of
vertices to grow the treewidth, so we propose the following
conjecture.

\begin{conjecture}
  \label{conj:log-treewidth-ehf}
  There exists a constant $c$ such that for any (even hole, $K_4$)-free graph $G$,
  the treewidth of $G$ is at most $c\log |V(G)|$.
\end{conjecture}

Our construction of even-hole-free layered wheels contains diamonds (a
\emph{diamond} is a graph obtained for $K_4$ by removing an edge).  We
therefore propose the following conjecture.

\begin{conjecture}
  \label{conj:diamond}
  Even-hole-free graphs with no $K_4$ and no diamonds have bounded
  treewidth. 
\end{conjecture}

(Even hole, pyramid)-free graphs attracted some attention 
(see~\cite{chudnovsky.t.t.v:ehfpyramidf}). 
It is therefore worth noting that even-hole-free layered wheels are pyramid-free
(see Theorem~\ref{th:layeredWheel-EHF_noPyramid}). 
We remark that it is also possible to obtain a variant of even-hole-free 
layered wheel that does contain pyramids. We omit giving all details about 
this construction that is still of interest because it might give indications 
about how an even-hole-free graph can be decomposed (or not) around a pyramid. 

We note that for the classes where we prove unbounded treewidth, the
cliquewidth (and therefore the rankwidth), to be defined later, is
also large (see Theorems~\ref{th:largeCW} and~\ref{th:lowerbound-rankwidth}).


\subsection*{Outline of the article}
 
In Section~\ref{sec:terminology}, we introduce the terminology used in
our proofs.

In Section~\ref{sec:construction}, we describe the construction of
layered wheels for two classes of graphs: (theta, triangle)-free
graphs and (even hole, $K_4$)-free graphs (in fact, we prove it for a
more restricted class namely (even hole, $K_4$, pyramid)-free graphs). We prove that the
constructions actually yield graphs in the corresponding classes (this is
non-trivial, see Theorems~\ref{th:layeredWheel-TTF},~\ref{th:layeredWheel-EHF},
and~\ref{th:layeredWheel-EHF_noPyramid}). We then prove that layered wheels have
unbounded treewidth (see Theorem~\ref{th:largeTW}) and cliquewidth
(see Theorem~\ref{th:largeCW}).

In Section~\ref{sec:lower-bound-rankwidth}, we recall the definition
of rankwidth. We exhibit (theta, triangle)-free graphs
and (even hole, $K_4$)-free graphs with large rankwidth. This is a
trivial corollary of Theorem~\ref{th:largeCW}, but the computation is
more accurate (see Theorem~\ref{th:lowerbound-rankwidth}).

In Section~\ref{sec:upper-bound}, we give an upper bound on the
treewidth of layered wheels. We prove a stronger result: the
so-called \emph{pathwidth} of layered wheels is bounded by some linear
function of the number of its layers (see Theorem~\ref{th:bounded-pw}).


\section{Summary of the main results and terminology}
\label{sec:terminology}

The treewidth, cliquewidth, rankwidth, and pathwidth of
a graph $G$ are denoted by $\tw(G)$, $\cw(G)$, $\rw(G)$, and $\pw(G)$
respectively.  The following lemma is well-known.

\begin{lemma}[See~\cite{CorneilR05} and~\cite{OumS06}]
  \label{lem:width-comparation}
  For every graph $G$, the followings hold:
  \begin{itemize}
  	\item $\rw(G) \leq \cw(G) \leq 2^{\rw(G)+1}$;
  	\item $\cw(G) \leq 3 \cdot 2^{\tw(G)} - 1$;
  	\item $\tw(G) \leq \pw(G)$.
  \end{itemize}
\end{lemma}


The first item of the lemma is proved in~\cite{CorneilR05}, and the second item is proved in~\cite{OumS06}.
The third item follows because pathwidth is a special case of treewidth (see Section~\ref{sec:upper-bound}). 
All results presented in this article can be summarized in the next two theorems.

\begin{theorem}
For every integers $l\geq 1$ and $k\geq 4$, there exists a graph
$G_{l, k}$ such that the followings hold:

\begin{itemize}
  \item $G_{l, k}$ is theta-free and has girth at least~$k$ (in particular, $G_{l, k}$ is triangle-free);
  \item $l \leq \tw(G_{l, k}) \leq \pw(G_{l,k}) \leq 2l$;
  \item $l \leq \rw(G_{l, k}) \leq \cw(G_{l, k}) \leq 3 \cdot 2^{\tw(G)} - 1 \leq 3 \cdot 2^{2l} - 1 \leq |V(G_{l, k})|$.
\end{itemize}

\end{theorem}

\begin{theorem}
For every integers $l\geq 1$ and $k\geq 4$, there exists a graph
$G_{l, k}$  such that the followings hold:

\begin{itemize}
  \item $G_{l, k}$ is (even hole, $K_4$, pyramid)-free and every hole in
  	$G_{l, k}$ has length at least~$k$;
  \item $l \leq \tw(G_{l, k}) \leq \pw(G_{l,k}) \leq 2l$;
  \item $l \leq \rw(G_{l, k}) \leq \cw(G_{l, k}) \leq 3 \cdot 2^{\tw(G)} - 1 \leq 3 \cdot 2^{2l} - 1 \leq |V(G_{l, k})|$.
  \end{itemize}
\end{theorem}

A graph $H$ is a {\em subgraph} of a graph $G$, denoted by
$H \seq G$, if $V(H) \seq V(G)$ and $E(H) \seq E(G)$. For a graph $G$
and a subset $X \seq V(G)$, we let $G[X]$ denote the subgraph of $G$
{\em induced} by $X$, i.e.\ $G[X]$ has vertex set $X$, and
$E(G[X])$ consists of the edges of $G$ that have both ends in $X$. 

For simplicity, sometimes we do not distinguish between a vertex set
and the graph induced by the vertex set.  So we write $G \sm H$
instead of $G[V(G) \sm V(H)]$. Also for a vertex $v \in V(G)$, we
write $G \sm v$ (instead of $G[V(G) \sm \{v\}]$) and similarly, we
write $G \sm S$ for some $S \seq V(G)$.  For $v \in V(G)$, we denote
by $N_H(v)$, the set of neighbors of $v$ in $H$ that is called the
{\em neighborhood} of $v$, and $N_G(v)$ is also denoted by $N(v)$.

A {\em path} in $G$ is a sequence $P$ of distinct vertices
$p_1 \dots p_n$, where for $i,j \in \{1,\dots,n\}$,
$p_ip_{j} \in E(G)$ if and only if $|i-j|= 1$.  For two vertices
$p_i, p_j \in V(P)$ with $j > i$, the path $p_ip_{i+1}\dots p_j$ is
a {\em subpath of $P$} that is denoted by $p_iPp_j$. The subpath
$p_2\dots p_{n-1}$ is called the {\em interior} of $P$. The vertices
$p_1,p_n$ are the \emph{ends} of the path, and the vertices in the
interior of $P$ are called the \emph{internal} vertices of $P$.
 
A {\em cycle} is defined similarly, with the additional properties
that $n\geq 4$ and $p_1 = p_n$.  The {\em length} of a path $P$ is the number of edges of $P$.  The
length of cycle is defined similarly.

We now give a formal definition of treewidth.
A {\em tree decomposition} of a graph $G$ is a pair 
$(T, \{X_t\}_{t \in V(T)})$, where 
$T$ is a tree whose every node $t$ is assigned a vertex
subset $X_t \subseteq V(G)$, called a {\em bag}, 
such that the following three conditions hold:

\begin{enumerate}[label=(T\arabic*)]
  \item $\bigcup_{t \in V(T)} X_t = V(G)$, 
    i.e., every vertex of $G$ is in at least one bag.
  \item For every $uv \in E(G)$, there exists a node $t$ of $T$ such that bag $X_t$
	contains both $u$ and $v$.
  \item For every $u \in V(G)$, the set $T_u = \{t \in V(T) : u \in X_t\}$, i.e., the set
	of nodes whose corresponding bags contain $u$, induces a connected subtree of $T$.
\end{enumerate}

The {\em width} of tree decomposition $(T, \{X_t\}_{t \in V(T)})$ equals
$\max_{t \in V(T)} |X_t| - 1$, that is, the maximum size of its bag minus 1. The {\em treewidth} of a graph $G$,
denoted by $\tw(G)$, is the minimum possible width of a tree decomposition of $G$.

\section{Construction and treewidth}
\label{sec:construction}
\label{sec:layeredWheels}

In this section, we describe the construction of layered wheels for
two classes of graphs, namely the class of (theta, triangle)-free
graphs and the class of (even hole, $K_4$)-free graphs. We also
give a lower bound on their treewidth.

\subsection*{(Theta, triangle)-free layered wheels}
\label{subsec:tff-LW}

We now present {\em ttf-layered-wheels} which are theta-free graphs
of girth at least~$k$, containing $K_{l+1}$ as a minor, for all integers
$l \geq 1, k\geq 4$ (see Figure~\ref{fig:ttf-layered-wheel}).

\begin{construction} \label{cons:G1(l,k)}

  Let $l \geq 0$ and $k \geq 4$ be integers. An
  $(l,k)$-ttf-layered-wheel, denoted by $G_{l,k}$, is a graph consisting
  of $l+1$ layers, which are paths $P_0, P_1,\dots, P_l$.  The graph
  is constructed as follows.
  	
\begin{enumerate}[label=(A\arabic*)]
\item \label{axi:A1} 
  $V(G_{l, k})$ is partitioned into $l+1$
  vertex-disjoint paths $P_0,P_1,\dots, P_l$.  So,
  $V(G_{l, k}) = V(P_0) \cup \cdots \cup V(P_l)$.  The paths are
  constructed in an inductive way.
  
\item 
  The path $P_0$ consists of a single vertex.

\item \label{axi:A3} 
  For every $0 \leq i \leq l$ and every vertex $u$ in
  $P_i$, we call \emph{ancestor of $u$} any neighbor of $u$ in
  $V(P_0) \cup \cdots \cup V(P_{i-1}) $.  The
  \emph{type of $u$} is the number of its ancestors (as we will see,
  the construction implies that every vertex has type~0 or 1).
  Observe that the unique vertex of $P_0$ has type~0. We will see that
  the construction implies that for every $1 \leq i \leq l$, the ends
  of $P_i$ are vertices of type~1.
 		
\item \label{axi:A4} 
  Suppose inductively that $l \geq 1$ and layers $P_0,P_1,\dots, P_{l-1}$ are
  constructed. The $l^{\textrm{th}}$-layer
  $P_{l}$ is built as follows.

  For any $u \in P_{l-1}$ we define a path
  $\myBox_u$ (that will be a subpath of $P_l$), in 
  the following way:
  
  \begin{itemize}
  \item if $u$ is of type~0, $\myBox_u$ contains three
    neighbors of $u$, namely $u_1, u_2, u_3$, in such way that
    $\myBox_u = u_1 \dots u_2 \dots u_3$.

  \item if $u$ is of type~1, let $v$ be its unique
    ancestor.  $\myBox_u$ contains six neighbors of $u$, namely
    $u_1,\dots, u_6$, and three neighbors of $v$, namely
    $v_1,v_2,v_3$, in such a way that $$\myBox_u = u_1 \dots u_2 \dots
    u_3 \dots v_1 \dots v_2 \dots v_3 \dots u_4 \dots u_5 \dots u_6.$$
  \end{itemize}

  The neighbors of $u$ and the neighbors of $v$ in $\myBox_u$
  are of type~1, the other vertices
  of $\myBox_u$ are of type~0. We now specify the lengths of the
  boxes and how they are connected to form $P_l$.
  		
\item 
  The path $P_{l}$ goes through the boxes of $P_{l}$ in the same
  order as vertices in $P_{l-1}$. For instance, if $uvw$ is a subpath
  of $P_{l-1}$, then $P_l$ goes through $\myBox_u$, $\myBox_v$, and
  $\myBox_w$, in this order along $P_{l}$. Note that the vertices of
  $P_l$ that are in none of the boxes are of type~0.  Note that for
  $u\neq v$, we have $\myBox_u \cap \myBox_v = \emptyset$. 

\item \label{axi:A6} 
  Let $w,w'$ be vertices of type~1 in $P_{l}$ (so
  vertices from the boxes), and consecutive in the sense that the
  interior of $wP_{l}w'$ contains no vertex of type~1. Then $wP_{l}w'$
  is a path of length at least~$k-2$.

\item 
  Observe that every vertex in $P_l$ has type~0 or 1.
  
\item \label{axi:A8} 
  There are no other vertices or edges apart from the ones specified
  above.  
    	
\end{enumerate}
\end{construction}

\begin{figure}[ht]
  \begin{center}
    \includegraphics[width=10cm]{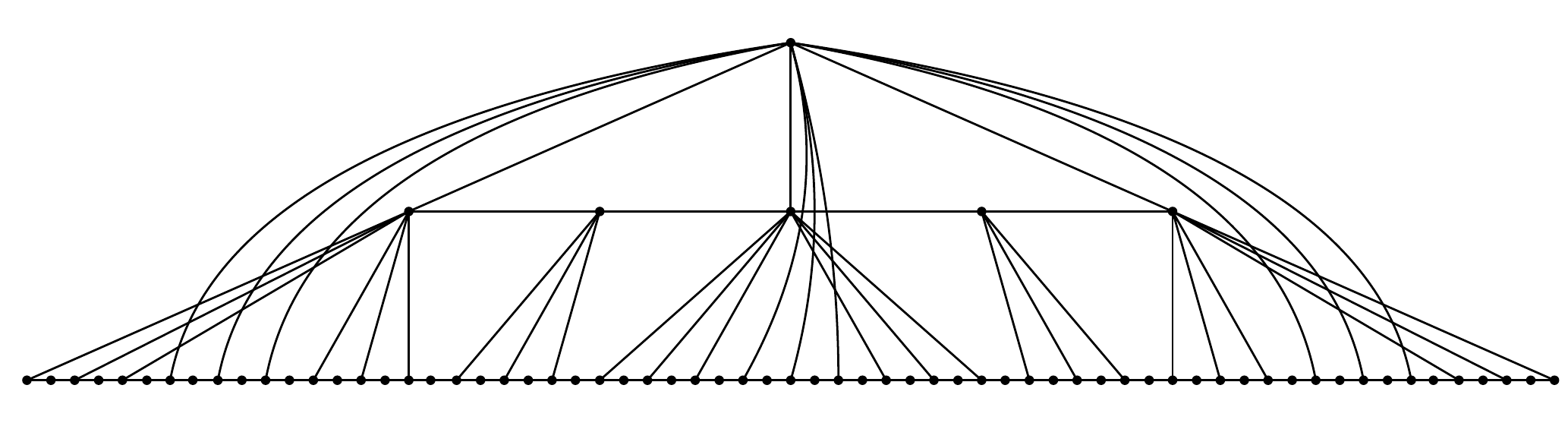}
  \end{center}
  \caption{A ttf-layered-wheel $G_{2,4}$ \label{fig:ttf-layered-wheel}}
\end{figure}

Observe that the construction is not fully deterministic because
in~\ref{axi:A6}, we just indicate a lower bound on the length of
$wP_lw'$, so there may exist different ttf-layered-wheels $G_{l, k}$.
This flexibility will be convenient below to exhibit
ttf-layered-wheels of arbitrarily large rankwidth.

\begin{lemma}
  \label{lem:neighborhood-ttf-layered-wheel}
  For $0 \leq i \leq l-1$ and $i+1 \leq j \leq l$, every vertex $u \in V(P_i)$
  has at least~$3^{j-i}$ neighbors in $P_j$.
\end{lemma}

\begin{proof}
  We prove the lemma by induction on $j$. If $j=i+1$, 
  then~\ref{axi:A4} implies that for every $0\leq i \leq l-1$ and every
  vertex $u$ in $P_i$, $u$ has three or six neighbors in $P_{i+1}$. If
  $j>i+1$, then by the induction hypothesis, every vertex
  $u\in V(P_i)$ has at least~$3^{j-1-i}$ neighbors in $P_{j-1}$. Hence
  by~\ref{axi:A4}, it has at least~$3 \cdot 3^{j-1-i} = 3^{j-i}$
  neighbors in $P_j$. 
\end{proof}

Lemma~\ref{lem:neighborhood-ttf-layered-wheel} implies in particular that
every vertex of layer $i$ has neighbors in all layers $i+1,\dots, l$.
Construction~\ref{cons:G1(l,k)} is in fact the description of an
inductive algorithm that constructs $G_{l, k}$. So, the next lemma is
clear.

\begin{lemma}
\label{lem:ttf-exist}
  For every integers $l \geq 0$ and $k \geq 4$, there exists an
  $(l,k)$-ttf-layered-wheel.
\end{lemma}

We now prove that Construction~\ref{cons:G1(l,k)} produces a theta-free 
graph with arbitrarily large girth and treewidth.
Observe that any subdivision of the (3,5)-grid contains a theta.
Thus, Theorem~\ref{th:layeredWheel-TTF} implies that a ttf-layered-wheel 
does not contain any subdivision of (3,5)-grid as mentioned in the introduction.

The next lemma is useful to prove Theorem~\ref{th:layeredWheel-TTF}.
For a theta consisting of three paths $P_1,P_2,P_3$, the common ends
of those paths are called the \emph{apexes} of the theta.  Let $G$ be
graph containing a path $P$.  The path $P$ is {\em special} if
\begin{itemize}
	\item there exists a vertex $v \in V(G \sm P)$ such that $|N_P(v)| \geq 3$; and
	\item in $G \sm v$, every vertex of $P$ has degree at most 2.
\end{itemize}
    
Note that in the next lemma, we make no assumption on $G$, that in
particular may contain triangles.

\begin{lemma}
  \label{lem:apex-in-path}
  Let $G$ be a graph containing a special path $P$.
  For any theta that is contained in $G$ (if any), every apex of
  the theta is not in $P$.
\end{lemma}

\begin{proof}
  Let $v$ be a vertex satisfying the properties as in the definition
  of special path.  For a contradiction, suppose that $P$ contains
  some vertex $u$ which is an apex of some theta $\Theta$ in $G$. Note
  that $u$ must have degree 3, and is therefore a neighbor of $v$.
  Consider two subpaths of $P$, $u_1Pu_2$ and $u_2Pu_3$ such that
  $u\in \{u_1, u_2, u_3\} \seq N(v)$ and both $u_1Pu_2$, $u_2Pu_3$
  have no neighbors of $v$ in their interior.  This exists since
  $|N_P(v)| \geq 3$.  Since $u$ is an apex, either $H_1 = vu_1Pu_2v$
  or $H_2 = vu_2Pu_3v$ is a hole of $\Theta$.  Without loss of
  generality suppose that $V(H_1) \seq V(\Theta)$. Hence the other
  apex of $\Theta$ must be also contained in $H_1$. Since
  $u_1v, u_2v \in E(G)$ and all vertices of $H_1 \sm \{u_1,v,u_2\}$
  have degree 2, $u_1, u_2$ must be the two apexes of $\Theta$. Since
  $d(u_2) = 3$, $V(u_2Pu_3) \seq \Theta$. But then $v$ has degree~3 in
  $\Theta$ while not being an apex, a contradiction. \cpl
\end{proof}

\begin{theorem}
  \label{th:layeredWheel-TTF}
  For every integers $l \geq 0$ and $k \geq 4$, every
  $(l, k)$-ttf-layered-wheel $G_{l,k}$ is theta-free
  graph with girth at least~$k$.
\end{theorem}

\begin{proof}
  \setcounter{claim}{0}
  
  We first show by induction on $l$ that $G_{l,k}$ has girth at least
  $k$.  This is clear for $l\leq 1$, so suppose that $l\geq 2$ and let
  $H$ be a cycle in $G_{l,k}$ whose length is less than~$k$.  We may
  assume that layer $P_l$ contains some vertex of $H$, for otherwise
  $H$ is a cycle in $G_{l-1, k}$, so it has length at least~$k$ by the
  induction hypothesis.  Let $P = u \dots v$ be a path such that
  $V(P) \seq V(H) \cap V(P_l)$ and with the maximum length among such
  possible paths.  Note that $P$ contains at least two vertices.
  Indeed, if $P$ contains a single vertex, then such a vertex must
  have at least two ancestors, since it has degree 2 in $H$, which is
  impossible by the construction of $G_{l,k}$. So $u \neq v$.
  Moreover, note that as $P$ is contained in a cycle, both $u$ and $v$
  must have an ancestor. Let $u'$ and $v'$ be the ancestor of $u$ and
  $v$ respectively. By~\ref{axi:A6} of Construction~\ref{cons:G1(l,k)}
  $P$ has length at least~$k-2$. Hence $u'uPvv'$ has length at least
  $k-1$, so $H$ has length at least~$k$. \cpl
  
  \vspace{1 ex}
  
  Now we show that $G_{l,k}$ is theta-free. For a
  contradiction, suppose that it contains a theta. Let $\Theta$ be a
  theta with minimum number of vertices, and having $u$ and $v$ as
  apexes.  As above, without loss of generality, we may assume that $P_l$ contains some
  vertex of $\Theta$. Note that every vertex of $P_l$ is contained in
  a special path of $G_{l, k}$.  Hence, by Lemma~\ref{lem:apex-in-path},
  $u,v \notin V(P_l)$.  In particular, every vertex of
  $V(P_l) \cap V(\Theta)$ has degree 2 in $\Theta$.

  Let $P = x \dots y$ for some $x, y \in P_l$, be a path such that
  $V(P) \subseteq V(\Theta) \cap V(P_l)$ and it is inclusion-wise
  maximal with respect to this property.  Since every vertex of $P_l$ has at
  most one ancestor, $x \neq y$.  Moreover, both $x$ and $y$ must have
  an ancestor, because every vertex of $\Theta$ has degree 2 or 3 in
  $\Theta$.  Let $x'$ and $y'$ be the ancestor of $x$ and $y$
  respectively. By the maximality of $P$, both $x'$ and $y'$ are also
  in $\Theta$. Note that no vertex in the interior of $P$ is adjacent
  to $x'$ or $y'$, since otherwise such a vertex would have degree 3
  in $\Theta$, meaning that it is an apex, a contradiction.
  
  \vspace{1 ex}

  \begin{claim}
    \label{cl:ancestors}
    We have $x' \neq y'$, $x'y' \notin E(G_{l, k})$, and
    some internal vertex of $P$ is of type~1.
  \end{claim}
  
  \bpc Otherwise, $x' = y'$ or $x'y' \in E(G_{l, k})$, or every
  internal vertex of $P$ is of type~0. In the last case, we also have
  $x' = y' \in V(P_{l-1})$ or $x'y' \in E(G_{l, k})$ by the
  construction of $G_{l, k}$. Hence, in all cases,
  $V(P) \cup \{x',y'\}$ induces a hole in $\Theta$, that must contain
  both $u$ and $v$. Since $u,v \notin V(P_l)$, we have
  $u,v \in \{x',y'\}$.  But this is not possible as $x' = y'$ or
  $x'y' \in E(G_{l, k})$. \epc
  
  \vspace{1 ex}
  
  We now set $P' = x'xP_lyy'$ (that is a path by
  Claim~\ref{cl:ancestors}).

  \vspace{1 ex}
  
  \begin{claim}
    \label{cl:property-of-P}  
    There exists no vertex of type~0 in $P_{l-1}$  that has a neighbor
    in the interior of $P$. 
  \end{claim}

  \bpc For a contradiction, let $t \in V(P_{l-1})$ be of type~0 that
  has neighbors in the interior of $P$. Note that
  $t \notin V(\Theta)$ because internal vertices of $P$ have degree~2 in
  $\Theta$. Let $Q$ be the shortest path from $x'$ to $y'$ in
  $G_{l, k}[V(P') \cup \{t\}]$.  Note that $Q$ is shorter than $P'$,
  because it does not go through one vertex of $N_P(t)$. So, $P'$ can
  be substituted for $Q$ in $\Theta$, which provides a theta from $u$
  to $v$ with less vertices, a contradiction to the minimality of
  $\Theta$. \epc
  
  \vspace{1 ex}

  \begin{claim}
    \label{cl:descriptionP}
    We may assume that:
    \begin{itemize}
    \item $x'\in V(P_{l-1})$ and $x'$ has type 0.
    \item $y'\notin V(P_{l-1})$.
    \item $y'$ has a neighbor $w$ in $P_{l-1}$ and $x'w\in E(G_{l, k})$.
    \item Every vertex in $P$ has type~0, except $x$, $y$, and three
      neighbors of $w$. Observe that $w$ has type~1 and has three more
      neighbors in $P_l$ that are not in~$P$.
    \end{itemize}
  \end{claim}

  \bpc Suppose first that $x', y'$ are both in $P_{l-1}$. Then by
  Claim~\ref{cl:ancestors}, the path $x'P_{l-1}y'$ has length at
  least two.  Moreover, by Claim~\ref{cl:property-of-P}, all its internal
  vertices are of type~1, because they all have neighbors in the
  interior of $P$.  It follows that $x'P_{l-1}y'$ has length
  exactly two. We denote by $z$ its unique internal vertex. Substituting
  $x'zy'$ for $P'$, we obtain a theta that contradicts the minimality
  of $\Theta$.  Observe that the ancestor of $z$ is not in
  $V(\Theta)$, because it has three neighbors in $P$. This proves that
  $x', y'$ are not both in $P_{l-1}$.

  So up to symmetry, we may assume that $y'\notin V(P_{l-1})$.  Since $y'$
  has neighbor in $P_l$, it must be that $y'$ has a neighbor
  $w\in V(P_{l-1})$, and that along $P_l$, one visits in order three
  neighbors of $w$, then $y$ and two other neighbors of $y'$, and then
  three other neighbors of $w$.

  Let $w'$ be the neighbor of $w$ in $P_{l-1}$, chosen so that $w'$
  has neighbors in $P$. Since $w'$ has type 0, by
  Claim~\ref{cl:property-of-P}, we have $w'=x'$.  Hence, as claimed,
  $x'\in V(P_{l-1})$ and $x'w\in E(G)$.  \epc

  \vspace{1 ex}
  
  Let $a$, $b$, $c$, $a'$, $b'$, $c'$ be the six neighbors of $w$ in
  $P_l$ appearing in this order along $P_l$, in such a way that
  $a, b, c \in V(P)$ and $a', b', c' \notin V(P)$. We have
  $\{a', b', c'\} \cap V(\Theta) \neq \emptyset$, since otherwise we
  obtain a shorter theta from $u$ to $v$ by replacing $P'$ with
  $x'wy'$, a contradiction to the minimality of $\Theta$. Let $y''$ be
  the neighbor of $y'$ in $yP_la'$ closest to $a'$ along $yP_la'$. Since
  $w \notin V(\Theta)$, $V(y' y'' P_l c') \seq V(\Theta)$.

  If $y'\notin \{u, v\}$, then by replacing $x' P' y' y'' P_l c'$ with
  $x' w c'$, we obtain a theta, a contradiction to the minimality of
  $\Theta$. So, $y'\in \{u, v\}$. Without loss of generality, we may assume that
  $y' = v$.

  If $u \neq x'$, then by replacing $V(x'P'y' y'' P_l c')$ with
  $\{x', w, y', c'\}$ in $\Theta$, we obtain a theta from $w$ to $u$
  which contains less vertices than $\Theta$, a contradiction to the
  minimality of $\Theta$. So, $u = x'$.
  
  Recall that $x'$ has type~0.  Let $z \neq w$ be the neighbor of $x'$
  in $P_{l-1}$. Moreover, let $z'$ and $z''$ be the neighbor of $z$
  and $x'$ in $P_l$ respectively, such that all vertices in the interior of
  $z' P_l z''$ have degree 2.  Since $\Theta$ goes through $P$,
  $w \notin V(\Theta)$. Therefore $z, z', z'' \in V(\Theta)$.  This
  implies the hole $z x' z'' P_l z' z$ is a hole of $\Theta$, a
  contradiction because the other apex $v = y'$ is not in the hole.
  This completes the proof that $G_{l,k}$ is theta-free.
\end{proof}

\subsection*{Even-hole-free layered wheels}

\label{subsec:layeredEHFWheels}

Recall that (even hole, triangle)-free graphs have treewidth at most 5
(see~\cite{CameronSHV18}), and as we will see,
ttf-layered-wheels of arbitrarily large treewidth exist.  Hence, some
ttf-layered-wheels contain even holes (in fact, it can be checked that
they contain even wheels).  We now provide a construction of layered
wheel that is (even hole, $K_4$)-free, but that contains triangles (see
Figure~\ref{fig:ehf-layered-wheel}).  Its structure is similar to
ttf-layered-wheel, but slightly more complicated. 

The construction of ehf-layered-wheel that we are going to discuss emerges from the 
structure of wheels that may exist in a graph of the studied class (namely, even-hole-free graphs
with no $K_4$).
In the class of even-hole-free graphs, a wheel may have several centers
while having the same rim. Those centers may be adjacent or not. In Figure~\ref{fig:ehf-LW-wheel},
we give examples of wheels that may exist in an even-hole-free graph.
Formally, we do not need to prove that these wheels are even-hole-free, and therefore we 
omit the (straightforward) proof.

\begin{figure}[ht]
  \centering \includegraphics[scale=0.5]{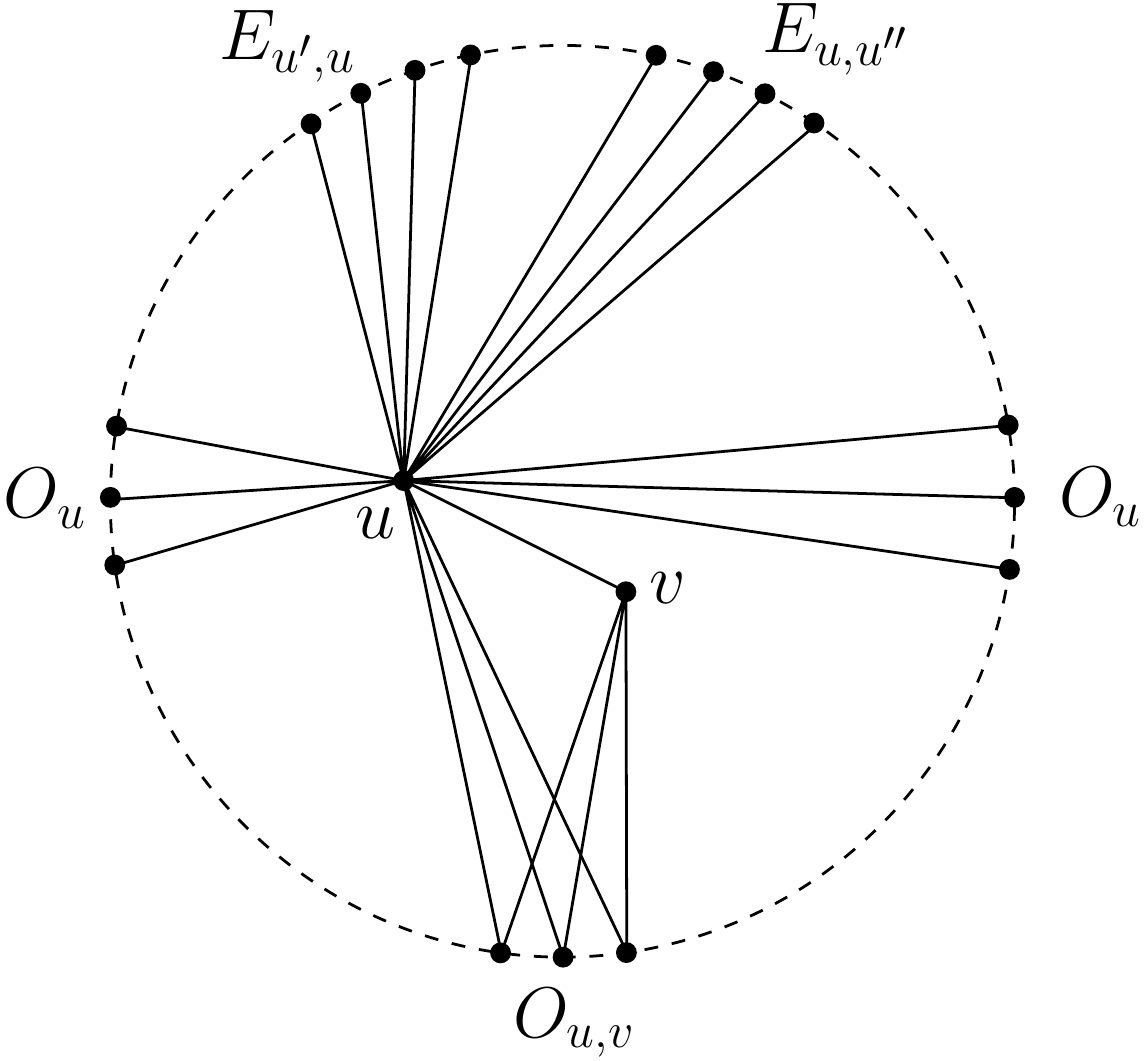} 
  \vspace{3ex}
  \centering \includegraphics[scale=0.5]{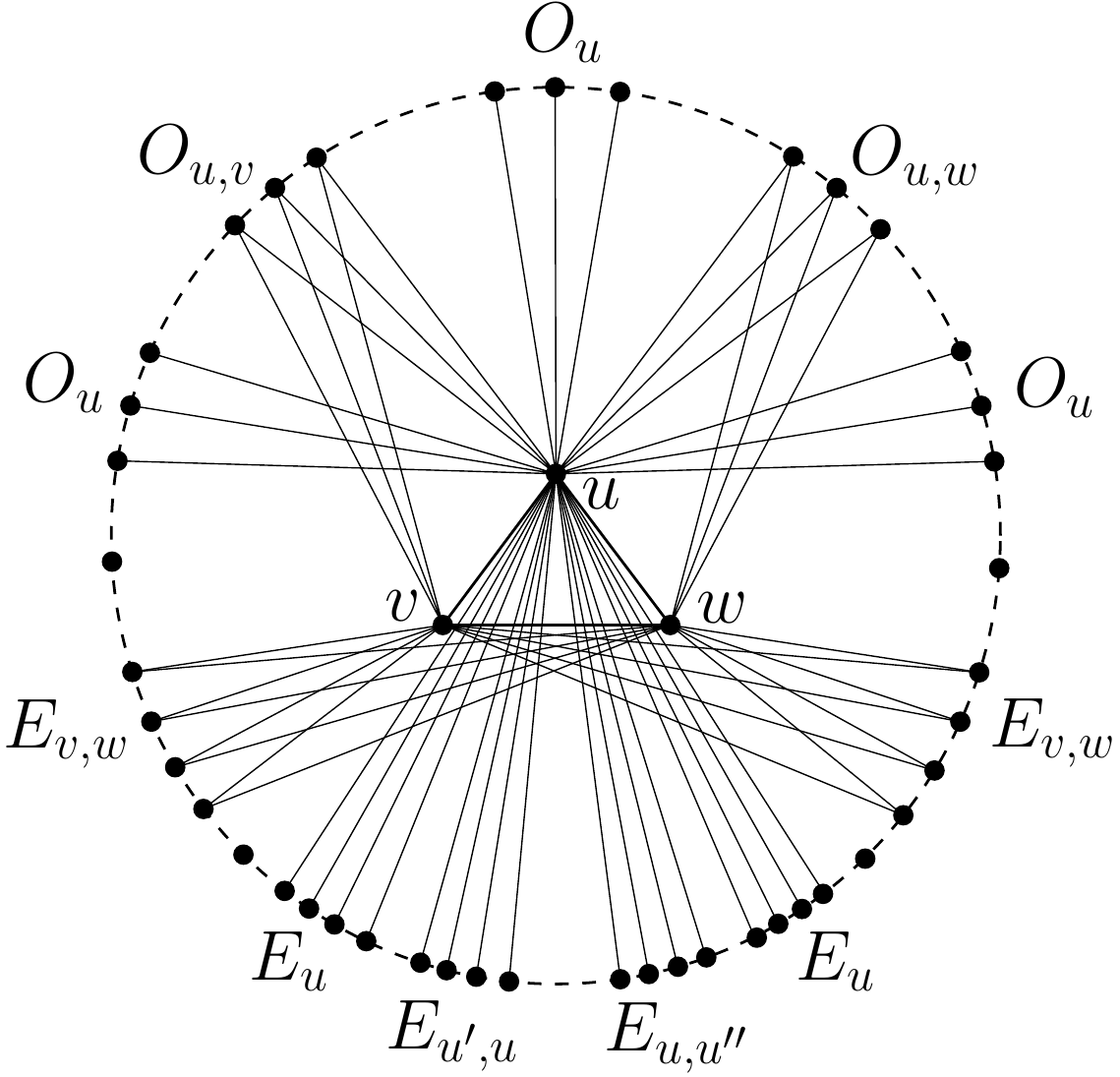}
  \caption{Wheels in an even-hole-free graph whose centers induce an edge or a triangle, 
  with the corresponding zones as described in Construction~\ref{cons:G2(l,k)} 
  (dashed lines between two vertices represent paths of odd length)}
  \label{fig:ehf-LW-wheel}
\end{figure}

Now we are ready to describe the construction of ehf-layered-wheel.

\begin{construction} \label{cons:G2(l,k)}

  Let $l \geq 1$ and $k \geq 4$ be integers. An
  $(l,k)$-ehf-layered-wheel, denoted by $G_{l,k}$, consists of $l+1$
  layers, which are paths $P_0, P_1,\dots, P_l$.  We view these paths
  as oriented from left to right. The graph is constructed as follows.
  	
\begin{enumerate}[label=(B\arabic*)]

\item \label{axi:B1} 
  $V(G_{l, k})$ is partitioned into $l+1$
  vertex-disjoint paths $P_0, P_1,\dots, P_l$. So,
  $V(G_{l, k}) = V(P_0) \cup \cdots \cup V(P_l)$.  The paths are
  constructed in an inductive way.
  
\item 
  The first layer $P_0$ consists of a single vertex $r$.  The
  second layer $P_1$ is a path such that
  $P_1= r_1P_1r_2P_1r_3$, where
  $\{r_1, r_2, r_3\} = N_{P_1}(r)$ and for $j = 1,2$,
  $r_jP_1r_{j+1}$ is of odd length at least~$k-2$.
    	
\item \label{axi:B3} 
  For every $0 \leq i \leq l$ and every vertex $u$ in
  $P_i$, we call \emph{ancestor of $u$} any neighbor of $u$ in
  $G_{l, k} \left[ P_0 \cup \cdots \cup P_{i-1} \right]$.  The
  \emph{type of $u$} is the number of its ancestors (as we will see,
  the construction implies that every vertex has type~0, 1, or 2).
  Observe that the unique vertex of $P_0$ has type~0, and $P_1$
  consists only of vertices of type~0 or type~1.  Moreover, we will
  see that if $u$ is of type~2, then its ancestors are adjacent. Also,
  the construction implies that for every $1 \leq i \leq l$, the ends
  of $P_i$ are vertices of type~1. 
    	
\item \label{axi:B4} 
  Suppose inductively that $l \geq 2$ and
  $P_0,P_1,\dots, P_{l-1}$ are constructed. The $l^{\textrm{th}}$-layer
  $P_{l}$ is built as follows.

  For all $0 \leq i \leq l-1$, any vertex $u \in V(P_i)$ has an odd
  number of neighbors in $P_{l}$, that are into subpaths of $P_l$ that
  we call \emph{zones}.  These zones are labeled by $E_u$ or $O_u$
  according to their parity: a zone labeled $E_u$ contains four
  neighbors of $u$, and a zone labeled $O_u$ contains three neighbors
  of $u$.  All these four or three neighbors are of type~1, and all
  the other vertices of the zone are of type~0.

  There are also zones that contain common neighbors of two vertices
  $u, v$. We label them $E_{u, v}$ (or $O_{u, v}$).  A zone $E_{u, v}$
  (resp.\ $O_{u, v}$) contains four (resp.\ three) common neighbors of
  $u$ and $v$.  All these four or three neighbors are of type~2, and
  all the other vertices of the zone are of type~0.

  The ends of a zone $E_u$ (resp.\ $O_u$) are neighbors of $u$. The
  ends of a zone $E_{u, v}$ (resp.\ $O_{u, v}$) are common neighbors
  of $u$ and $v$. Distinct zones are disjoint.

\item \label{axi:B5}  
  For any $u \in P_{l-1}$, we define the box $\myBox_u$, 
  that is a subpath of $P_l$, as follows:
  		
  \begin{itemize}
  \item If $u$ is of type~0 (so it is an internal vertex of
    $P_{l-1}$), then let $u'$ and $u''$ be the neighbors of $u$ in
    $P_{l-1}$, so that $u'uu''$ is a subpath of $P_{l-1}$.  In this
    case, $\myBox_u$ goes through three zones $E_{u', u}$, $O_u$,
    $E_{u, u''}$ that appear in this order along $P_{l}$ (see
    Figure~\ref{fig:ehf-layered-wheel}).

  \item  If $u$ is of type~1, then let $v \in P_i$,
    $i < l-1$ be its ancestor.

    If $u$ is an internal vertex of $P_{l-1}$, then let $u'$ and $u''$
    be the neighbors of $u$ in $P_{l-1}$, so that $u'uu''$ is a
    subpath of $P_{l-1}$.  In this case, $\myBox_u$ is made of five
    zones $E_{u', u}$, $O_u$, $O_{u, v}$, $O_u$, $E_{u, u''}$ (see
    Figure~\ref{fig:ehf-layered-wheel}).

    If $u$ is the left end of $P_{l-1}$, then let $u''$ be the
    neighbor of $u$ in $P_{l-1}$. In this case, $\myBox_u$ is made of
    four zones $O_u$, $O_{u, v}$, $O_u$, $E_{u, u''}$.

    If $u$ is the right end of $P_{l-1}$, then let $u'$ be the
    neighbor of $u$ in $P_{l-1}$. In this case, $\myBox_u$ is made of
    four zones $E_{u', u}$, $O_u$, $O_{u, v}$, $O_u$.
    		
  \item If $u$ is of type~2 (so it is an internal vertex
    of $P_{l-1}$), then let $v \in P_i$ and $w \in P_{j}$, $j \leq i$
    be its ancestors.  If $i=j$, we suppose that $v$ and $w$ appear in
    this order along $P_i$ (viewed from left to right).  It turns out
    that either $w$ is an ancestor of $v$, or $v, w$ are consecutive
    along some path $P_{i}$ (because as one can check, all vertices of
    type~2 that we create satisfy this statement).  In this case,
    $\myBox_u$ is made of 11 zones, namely $E_{u', u}$, $E_u$, $E_{v, w}$, $O_u$,
    $O_{u, v}$, $O_{u}$, $O_{u, w}$, $O_u$, $E_{v, w}$, $E_u$, and
    $E_{u, u''}$ (see Figure~\ref{fig:ehf-layered-wheel}). 
    		
  \end{itemize}
  
  Note that for any two adjacent vertices $u,v \in P_{l-1}$, 
  $\myBox_u$ and $\myBox_v$ are not disjoint.
  		
\item \label{axi:B6} 
  The path $P_{l}$ visits all the boxes
  $\myBox_{-}$ of $P_{l}$ in the same order as vertices in
  $P_{l-1}$. For instance, if $uvw$ is a subpath of $P_{l-1}$, then
  $\myBox_u$, $\myBox_v$, and $\myBox_w$ appear in this order along
  $P_{l}$.

\item \label{axi:B7} 
  Let $u$ and $v$ be two vertices of $P_l$, both of
  type~1 or 2, and consecutive in the sense that every vertex in the
  interior of $uP_{l}v$ is of type~0.  If $u$ and $v$ have a common
  ancestor, then $uP_{l}v$ has odd length, at least~$k-2$. If $u$ and
  $v$ have no common ancestor, then $uP_{l}v$ has even length, at
  least~$k-2$.

\item 
  Observe that every vertex in $P_l$ has type~0, 1, or 2. Moreover, as
  announced, every vertex of type~2 has two adjacent ancestors. 
  
\item \label{axi:B9} 
  There are no other vertices or edges apart from the ones
  specified above.
  
\end{enumerate}  	
\end{construction}

\begin{figure}[ht]
  \begin{center}
    \includegraphics[width=10cm]{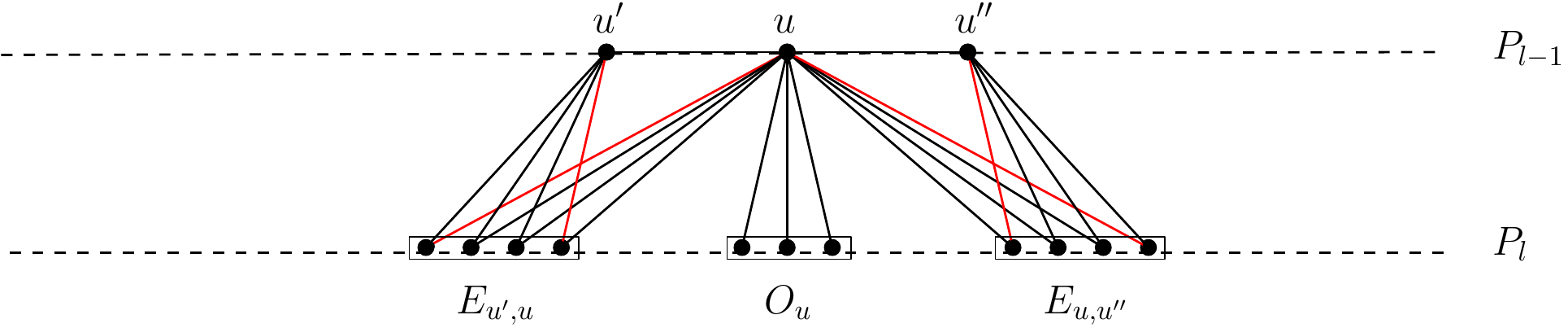}
  \end{center}
  \vspace{1ex} 
  \begin{center}
    \includegraphics[width=10cm]{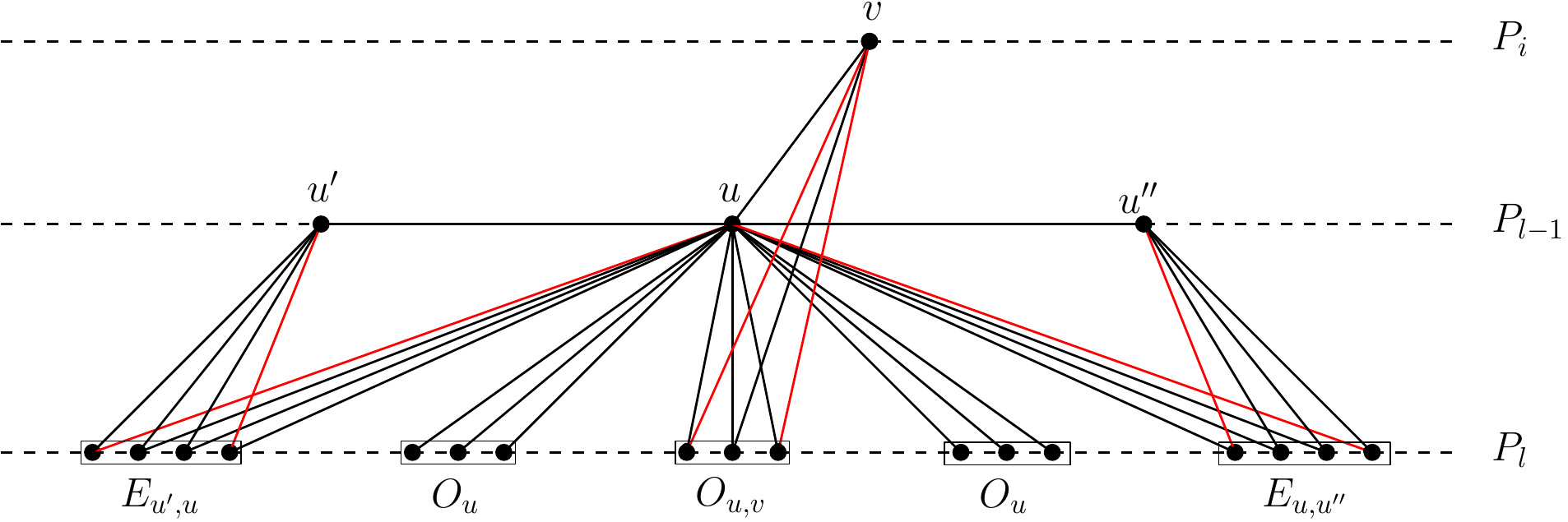}
  \end{center}
  \vspace{1ex} 
  \begin{center}
    \includegraphics[width=10cm]{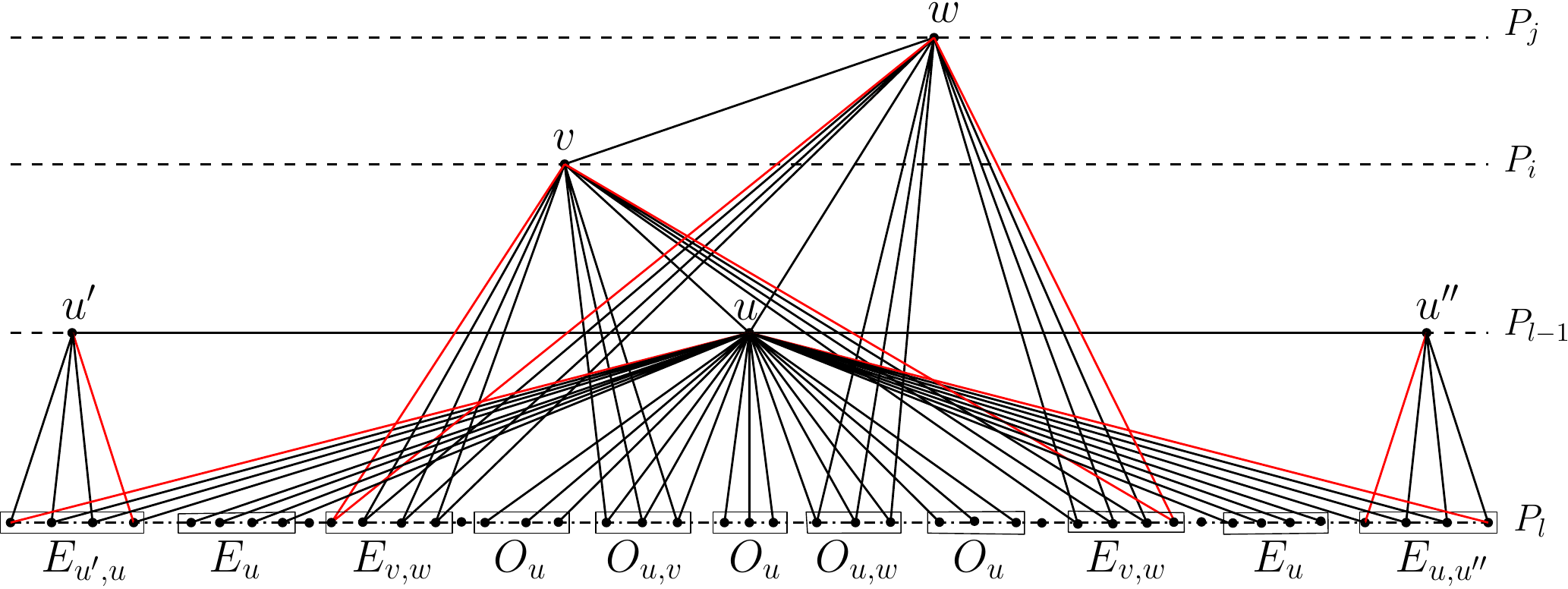}
  \end{center}
  \caption{The neighborhood of a type~0, type~1, or type~2 vertex 
  $u \in V(P_{l-1})$ in $P_l$ (dashed lines between two vertices in $P_l$ represent paths of odd length, and
  red edges represent non-internal edges as in the proof of Theorem~\ref{th:layeredWheel-EHF})
  \label{fig:ehf-layered-wheel}}
\end{figure}

For the same reason as for ttf-layered-wheels, we allow flexibility in
Construction~\ref{cons:G2(l,k)}, by just giving lower bounds for the
lengths of paths described in~\ref{axi:B7}.  So there may exist different
ehf-layered-wheels $G_{l, k}$ for the same value of $l$ and~$k$.

\begin{lemma}
  \label{lem:neighborhood-ehf-layered-wheel}
  For $0 \leq i \leq l-1$ and $i+1 \leq j \leq l$, every vertex $u \in V(P_i)$
  has at least~$3^{j-i}$ neighbors in $P_j$.
\end{lemma}

\begin{proof}
  We omit the proof since it is similar to the proof 
  of~Lemma~\ref{lem:neighborhood-ttf-layered-wheel}.
\end{proof}

\vspace{1 ex} Lemma~\ref{lem:neighborhood-ehf-layered-wheel} implies
that every vertex of layer $i$ has neighbors in all layers
$i+1,\dots, l$.  The next lemma is clear.

\begin{lemma}
  For every integers $l \geq 1$ and $k \geq 4$, there exists an
  $(l,k)$-ehf-layered-wheel.
\end{lemma}

\vspace{1 ex} We need some properties of lengths of some paths in
ehf-layered-wheel. It is convenient to name specific subpaths of boxes
first (see Figure~\ref{fig:ehf-layered-wheel}). 

\begin{itemize}
\item
Suppose that $u$ is a vertex in $P_{l-1}$ (of any type).

If $u$ is not an end of $P_{l-1}$, then a subpath of $\myBox_u$ is a
\emph{shared part of $\myBox_u$} if it is either the zone $E_{u', u}$
or the zone $E_{u, u''}$.  The \emph{private part of $\myBox_u$} is
the path from the rightmost vertex of $E_{u', u}$ to the leftmost
vertex of $E_{u, u''}$.

Otherwise, if $u$ is the left end of $P_{l-1}$ (and therefore of
type~1), then $u$ has only one shared part, that is the zone
$E_{u, u''}$, where $u'' \in N_{P_{l-1}}(u)$.  The private part of $u$
is the path from the leftmost vertex of the leftmost zone $O_u$ 
to the leftmost vertex of $E_{u, u''}$.

Similarly, if $u$ is the right end of $P_{l-1}$, then $u$ has only one
shared part, that is the zone $E_{u', u}$, where
$u' \in N_{P_{l-1}}(u)$.  The private part of $u$ is the path from the
rightmost vertex of $E_{u', u}$ to the rightmost vertex of the
rightmost zone $O_u$.

Observe that $\myBox_u$ is edgewise partitioned into a private part
and some shared parts (namely zero if $l=1$ and $u$ is the unique vertex of
layer $P_0$, one if $l>1$ and $u$ is an end of $P_{l-1}$, two
otherwise).

\item 
  Suppose that $u$ is of type~1 and $v$ is its ancestor.

  If $u$ is not the left end of $P_{l-1}$, then the \emph{left escape
    of $v$ in $\myBox_u$} is the subpath of $\myBox_u$ from the
  rightmost vertex of $E_{u', u}$ to the leftmost vertex of
  $O_{u, v}$.

  If $u$ is not the right end of $P_{l-1}$, then the \emph{right
    escape of $v$ in $\myBox_u$} is the subpath of $\myBox_u$ from the
  rightmost vertex of $O_{u, v}$ to the leftmost vertex of
  $E_{u, u''}$.

\item Suppose that $u$ is of type~2 and $v, w$ are its ancestors as in
  Construction~\ref{cons:G2(l,k)}.  Note that $u$ is not an end of
  $P_{l-1}$.

  The \emph{left escape of $v$ (resp.\ of $w$) in $\myBox_u$} is the
  subpath of $\myBox_u$ from the rightmost vertex of $E_{u', u}$ to
  the leftmost vertex of the zone $E_{v, w}$ that is the closest to $E_{u', u}$.

  The \emph{right escape of $v$ (resp.\ of $w$) in $\myBox_u$} is the
  subpath of $\myBox_u$ from the rightmost vertex of the zone
  $E_{v, w}$ that is the closest to $E_{u, u''}$, to the leftmost vertex of
  $E_{u, u''}$.
\end{itemize}


\begin{lemma} \label{lem:length-of-boxes}
    Let $G_{l, k}$ be an ehf-layered-wheel with $l\geq 1$ and $u$
    be a vertex in the layer $P_{l-1}$.  Then the following hold:
    \begin{itemize}
      \item Shared parts of $\myBox_u$ are paths of odd length. 
      \item The private part of $\myBox_u$ is a path of even length
      if $u$ is not an end of $P_{l-1}$; and it is of odd length
      otherwise.
      \item If $u$ has type 1 or 2, then all the left and right
        escapes of its ancestors in $\myBox_u$ are paths of even length.  
    \end{itemize}
\end{lemma}

\begin{proof}
  To check the lemma, it is convenient to follow the path $\myBox_u$
  on Figure~\ref{fig:ehf-layered-wheel} from left to right.
  Along this proof, we refer to Construction~\ref{cons:G2(l,k)},
  and we follow the notation given in Figure~\ref{fig:ehf-layered-wheel}.

  By~\ref{axi:B7}, shared parts of $\myBox_u$ have obviously odd length.
  
  If $u$ has type~0, then along the private part of $\myBox_u$, one
  meets 1 common neighbor of $u$ and $u'$, then 3 private neighbors of
  $u$, and then 1 common neighbor of $u$ and $u''$.  In total, from
  the leftmost neighbor of $u$ to its rightmost neighbor, one goes
  through 4 subpaths of $\myBox_u$, each of odd length by~\ref{axi:B7}
  (2 of the paths are in zones, while 2 of them are between
  zones). The private part of $\myBox_u$ has therefore even length.

  If $u$ has type~1, then the proof is similar.
  If it is not an end of $P_{l-1}$, then along the private part 
  of $\myBox_u$, one visits 10 subpaths
  (6 in zones, 4 between zones), each of odd length by~\ref{axi:B7}.
  Otherwise, one visits 9 subpaths
  (6 in zones, 3 between zones), each of odd length by~\ref{axi:B7}.

  If $u$ has type 2 then $u$ is not an end of $P_{l-1}$.
  Now there are more details to check. 
  Along the private part of $\myBox_u$, one visits 32
  subpaths. Among them, 22 are in zones and have odd length
  by~\ref{axi:B7}, and 10 are between zones. But 4 of the subpaths
  between zones have even length by~\ref{axi:B7}, namely, the paths
  linking $E_u$ to $E_{v, w}$ (because
  $\{u\} \cap \{v, w\} = \emptyset$), $E_{v,w}$ to $O_u$, 
  $O_{u}$ to $E_{v,w}$, and $E_{v,w}$ to $E_u$.  
  The 6 remaining subpaths between zones have odd
  length by~\ref{axi:B7}.  In total, the private part of $\myBox_u$ has even
  length as claimed.

  For the left and right escapes, the proof is similar.
  If $u$ is of type~1, then the escape is made of 4 paths each of odd length.
  If $u$ is of type~2, then the escape is made of 
  the path between zones $E_{v,w}$ and $E_u$ that is of even length,
  three paths in zone $E_u$ each of an odd length, and the path
  between zone $E_u$ and $E_{u',u}$ or $E_{u,u''}$ that is of odd length.
  So, every left and every right escape is of even length.
\end{proof}

\vspace{1 ex}


\begin{theorem} \label{th:layeredWheel-EHF}
  For every integers $l \geq 1$ and $k \geq 4$, every
  $(l, k)$-ehf-layered-wheel $G_{l,k}$ is (even hole, $K_4$)-free and
  every hole in $G_{l,k}$ has length at least~$k$.
\end{theorem}

\begin{proof}
  \setcounter{claim}{0}
  
  It is clear from the construction that $G_{l,k}$ does not
  contain~$K_4$.  Moreover, it follows from~\ref{axi:B7} that apart from
  triangles, any chordless cycle in $G_{l,k}$ is of length at least~$k$ (we omit
  the formal proof that is similar to the proof that
  ttf-layered-wheels have girth at least~$k$).

  For a contradiction, consider an ehf-layered-wheel $G_{l, k}$ that
  contains an even hole $H$. Suppose that $l$ is minimal, and under
  this assumption that $H$ has minimum length.  Hence, layer $P_l$
  contains some vertex of $H$, for otherwise
  $G_{l, k}[P_0 \cup \cdots \cup P_{l-1}]$ would be a counterexample.
  Let us start by the following claim.
    \vspace{1ex}
  
  Let $x$ be a vertex in $P_{i}$ where $0\leq i<l$, and $y$ be a
  neighbor of $x$ in $P_{l}$.  We say that $xy$ is an \emph{internal
    edge} (see Figure~\ref{fig:ehf-layered-wheel}) if one of the
  following holds:
  \begin{itemize}
    \item $i=l-1$ and $y$ is an internal vertex of $\myBox_x$.
    \item $i<l-1$, $x$ is an ancestor of $x'\in V(P_{l-1})$, $x'$ has
    type~1 or~2, $y$ is in $\myBox_{x'}$ and $y$ is neither the
    leftmost neighbor of $x$ in $\myBox_{x'}$ nor the rightmost
    neighbor of $x$ in $\myBox_{x'}$.
  \end{itemize}

  \begin{claim}
    \label{cl:intEdge}
    $H$ contains no internal edge. 
  \end{claim}
  
  \bpc 
Let $xy$ be an internal edge as in the definition and suppose for a
contradiction that $xy$ is an edge of $H$.  Let $Q=y\dots z$ be the
path of $H$ that is included in $P_l$ and that is maximal with respect
to this property.  Let $z'$ be the ancestor of $z$ that is in $H$ (it
exists by the maximality of $Q$).

Suppose first that $x$ is in $P_{l-1}$. We then set $x=u$ and observe that
$u$ has type~0,~1 or~2 (see Figure~\ref{fig:ehf-layered-wheel}).  If $u$ has type~0, then
since $uy$ is internal edge, $y$ is either in the zone $O_u$, or is among
the three rightmost vertices of zone $E_{u', u}$, or is among the
three leftmost vertices of zone $E_{u, u''}$ (where $u'$ and $u''$ are the left
and the right neighbors of~$u$ respectively in~$P_{l-1}$ as shown in 
Figure~\ref{fig:ehf-layered-wheel}). Since no internal vertex
of $Q$ is adjacent to $u$ because $H$ is a hole, we have $zu\in E(G)$
and $z'=u$.  So, $H=uyQzu$ and $H$ has odd length by the axiom~\ref{axi:B7}, a
contradiction.  If $u$ has type~1 (and ancestor $v$ as represented in
Figure~\ref{fig:ehf-layered-wheel}), the proof is similar (note in this case that $z'\neq v$ for
otherwise the triangle $uvy$ would be in $H$, a contradiction).

If $u$ has type~2 (and ancestors $v, w$ as represented in Figure~\ref{fig:ehf-layered-wheel})
the proof is similar with some additional situations.  For instance,
it can be that $y$ is the rightmost vertex of the leftmost zone
$Z=E_u$.  In this case, $z$ can be either the leftmost vertex of the
zone $E_{v, w}$ that is next to $Z$, or the leftmost vertex of the zone
$O_u$ that is closest to $Z$.  In the first case, $z'=v$ or $z'=w$
(say $z'=v$ up to symmetry), so $H=uyQzvu$ and $H$ has odd length by~\ref{axi:B7}; 
in the second case, $H=uyQzu$ and $H$ has again odd length by~\ref{axi:B7}, 
a contradiction.  Similar situations are when $y$ is the
leftmost vertex of the leftmost zone $O_u$, when $y$ is the rightmost
vertex of the rightmost zone $O_u$ and when $y$ is the leftmost vertex
of the rightmost zone $E_u$. We omit the details of each situation. 

Suppose now that $x$ is not in $P_{l-1}$.  Since $x$ has neighbor in
$P_l$, $x$ is the ancestor of some vertex $u$ from $P_{l-1}$.  If $u$ is
of type~1 with ancestor $v$, then $x=v$. We observe that $y$ must be
the middle vertex of the zone $O_{u, v}$. Hence, $H=vyQzv$ and $H$ has
odd length by~\ref{axi:B7}, a contradiction.

So, $u$ has type~2 and ancestors $v, w$.  Up to symmetry, we may
assume that $x=v$.  As in the previous cases, whatever the place of
$y$ in~$\myBox_u$, we must have either $H=vyQzv$, or $H=vyQzuv$, or $H=vyQzwv$
(when $y$ is the rightmost vertex of $O_{u, v}$ and $z$ is the leftmost vertex of $O_{u, w}$).
In all cases, $H$ has odd length, a contradiction.
\epc  

\vspace{1 ex}

  Now let $P = s \dots t$ be a subpath of $H$ in $P_l$ such that $P$ is
  inclusion-wise maximal.  So both $s$ and $t$ have an ancestor that is in~$H$.  If
  $P$ contains a single vertex (i.e., $s=t$), then $s$ must have two
  ancestors, say, $s_1$ and $s_2$, which are adjacent by \ref{axi:B3} of
  Construction~\ref{cons:G2(l,k)}.  Thus $\{s,s_1,s_2\}$ forms a triangle
  in $H$, which is not possible.  So $P$ contains at least two
  vertices and $s \neq t$.  Let $u$ and $v$ be ancestors of $s$ and~$t$ 
  respectively, such that $u,v \in V(H)$ (possibly $u = v$, or
  $uv\in E(G)$).

  Recall that all layers are viewed as oriented from left to right. We
  suppose that $s$ and $t$ appear in this order, from left to right,
  along $P_l$.
  
  \vspace{1 ex}
  
  \begin{claim} \label{cl:not-all-in-P}
    For every vertex $p \in V(P_{l-1})$, $N(p) \cap V({P_l}) \not\seq V(P)$.
  \end{claim}
  
  \bpc 
  Suppose that $p \in V(P_{l-1})$ and
  $N(p) \cap V({P_l}) \seq V(P)$.  So, $p\notin V(H)$.  Note that $p$
  is an internal vertex of $P_{l-1}$, for otherwise, $s$ or $t$ is an
  end of $P_l$ and has degree~2, while having two neighbors in
  $V(H) \cap V(G_{l, k}\sm p)$, a contradiction.

  By~\ref{axi:B5}, ancestors of $p$ (if any) and the neighbors of $p$
  in $P_{l-1}$ must also have neighbors in $P$. Thus, all of such
  vertices do not belong to $H$ because $P$ is a subpath of~$H$. By
  Lemma~\ref{lem:length-of-boxes}, the path $\myBox_p = p'\dots p''$
  has an even length.  Indeed $\myBox_p$ consists of two shared parts
  (each of odd length) and one private part (of even length).  It
  yields that $\myBox_p$ and $p'pp''$ have the same parity, and hence
  replacing $\myBox_p$ in $H$ with $p'pp''$ yields an even hole with
  length strictly less than the length of $H$, a contradiction to the
  minimality of $H$. 
  \epc
  
  \vspace{1ex}
  
  \begin{claim} \label{cl:where-is-u,v}
    Exactly one of $u$ and $v$ is in $P_{l-1}$.
  \end{claim}
  
  \bpc 
  Suppose that both $u$ and $v$ are not in $P_{l-1}$. Since $u$
  and $v$ have neighbors in $P$, each of them has a neighbor in
  $P_{l-1}$ (where such neighbors also have some neighbor in $P$).  Let
  $u'$ and $v'$ be the respective neighbors of $u$ and $v$ in $P_{l-1}$.
  
  If $u' = v'$, then $u'$ is a type~2 vertex in~$P_{l-1}$, with ancestors 
  $u$ and $v$, hence $u$ and $v$ are adjacent. 
  So, $H$ is a hole of form $usPtvu$, and it has odd length by construction.
  Therefore $u' \neq v'$, and by construction, the interior of
  $u'P_{l-1}v'$ must contain a vertex $w$ of type~0. It yields that
  $N_{P_l}(w)$ is all contained in $P$, a contradiction to
  Claim~\ref{cl:not-all-in-P}.

  Suppose now that both $u$ and $v$ are in $P_{l-1}$.  By
  Claim~\ref{cl:not-all-in-P}, no vertex of $P_{l-1}$ has all its
  neighbors in $P$.  So the interior of $uP_{l-1}v$ contains at most
  two vertices.

  If $u=v$, then by~\ref{axi:B7} $P$ is of odd length, and since
  $V(H) = \{u\}\cup V(P)$, $H$ is also of odd length, a contradiction.
  Similarly if $uv\in E(G)$, then by~\ref{axi:B7}, $P$ is of even length,
  $V(H) = \{u, v\}\cup V(P)$, and $H$ has odd length, again a contradiction.
  
  If the interior of $uP_{l-1}v$ contains a single vertex, then let
  $w$ be this vertex.  Let $w_1$ (resp.\ $w_2$) be the neighbor of $w$
  in $P$ that is closest to $s$ (resp.\ $t$).  Note that by
  \ref{axi:B5}, $s=w_1, t=w_2$ because both $u$ and $v$ are adjacent to
  $w$ in $P_{l-1}$. So, $sPt$ is the private part of~$\myBox_w$, and
  by Lemma~\ref{lem:length-of-boxes}, it has even length, as
  $uwv$. Moreover, by Claim~\ref{cl:intEdge},
  $\{s\} = V(E_{u, w}) \cap V(H)$ and $\{t\} = V(E_{w, v}) \cap V(H)$.
  Also, if $w$ has an ancestor, then such an ancestor must have
  neighbors in $P$, and hence it does not belong to $H$.  Altogether,
  we see that $N_H(w)\subseteq V(usPtv)$.  So, replacing $usPtv$ in
  $H$ with $uwv$ returns an even hole with length strictly less than
  the length of $H$, a contradiction to the minimality of $H$.

  So the interior of $uP_{l-1}v$ contains two vertices.  We let
  $uP_{l-1}v = uww'v$, and $w_1$ (resp.\ $w'_2$) be the neighbor of
  $w$ (resp.\ $w'$) in $P$ that is closest to $s$ (resp.\
  $t$). By~\ref{axi:B5}, $s=w_1, t=w'_2$.  So, $sPt$ is edgewise
  partitioned into the private part of $w$, the part shared between
  $w$ and $w'$, and the private part of~$w'$.  By
  Lemma~\ref{lem:length-of-boxes}, $sPt$ has therefore odd length. In
  particular, the length of $usPtv$ has the same parity as the length
  of $uww'v$.  Moreover, by Claim~\ref{cl:intEdge},
  $\{s\} = V(E_{u, w}) \cap V(H)$ and
  $\{t\} = V(E_{w', v}) \cap V(H)$.  Also, if $w$ or $w'$ has an
  ancestor, then such an ancestor must have neighbors in $P$, and
  hence it does not belong to $H$.  Altogether, we see that
  $N_H(\{w, w'\})\subseteq V(usPtv)$.  So, replacing $usPtv$ in $H$
  with $uww'v$ returns an even hole that is shorter than $H$, again a
  contradiction to the minimality of $H$.  
  \epc

\vspace{1 ex}

By Claim~\ref{cl:where-is-u,v} and up to symmetry, we may assume that
$u \in V(P_{l-1})$ and $v \notin V(P_{l-1})$.  So, $v$ has a neighbor
$v'$ in $P_{l-1}$ such that $t\in \myBox_{v'}$. Note that $v' \notin H$, 
because by construction $v'$ has some neighbor in~$P$.  Hence, $v' \neq u$
(because $u \in H$). If the path $uP_{l-1}v'$ 
has length at least three, then some vertex in the interior of
$uP_{l-1}v'$ contradicts Claim~\ref{cl:not-all-in-P}.

If $uP_{l-1}v'$ has length two, so $uP_{l-1}v' = uwv'$ for some vertex
$w \in V(P_{l-1})$, then $w$ is of type~0 because $v'$ is not of
type~0.  Hence, $P$ is edgewise partitioned into the private part of
$w$, the part of $\myBox_{v'}$ shared between $w$ and $v'$ and the
left escape of $v$ in $\myBox_{v'}$.  Let $w'$ be the rightmost vertex
of the shared zone $E_{w, v'}$. By Lemma~\ref{lem:length-of-boxes},
$usPw'$ has even length, as $uww'$.  Moreover, by
Claim~\ref{cl:intEdge}, $\{s\} = V(E_{u, w}) \cap V(H)$ and since $w$
has type~0, we see that $N_H(w)\subseteq V(usPw')$.  So, replacing
$usPw'$ in $H$ with $uww'$ returns an even hole with length strictly
less than the length of $H$, a contradiction to the minimality of $H$.
  
Hence, $uP_{l-1}v'$ has length one: $uP_{l-1}v' = uv'$. So, $P$ is the
left escape of $v$ in $\myBox_{v'}$.  By
Lemma~\ref{lem:length-of-boxes}, $P$ has even length.  By
Claim~\ref{cl:intEdge}, $\{s\} = V(E_{u, v'}) \cap V(H)$. Recall that
$v'\notin H$.  If $N_H(v') \subseteq V(usPtv)$, then replacing $usPtv$
in $H$ with $uv'v$ returns an even hole with length strictly less than
the length of $H$, a contradiction to the minimality of $H$.

So, $v'$ has neighbors in $H$ that are not in $usPtv$.  Note that if
$v'$ is of type~2, the ancestor of $v'$ that is different from $v$ is
not in $H$ (because it is adjacent to $t$ and to $v$).  Also, by
Claim~\ref{cl:intEdge}, the neighbors of $v'$ in $E_{u, v'}\sm s$ are not
in $H$. 

We denote by $v''$ the right neighbor of $v'$ in $P_{l-1}$. Note that
$v''$ has type~0, since $v'$ has type~1 or~2.  Let $s'$ and $t'$ be
vertices such that $t'P_ls'$ is the right escape of $v$ in
$\myBox_{v'}$, $t'$ is adjacent to $v$, and $s'$ is adjacent to $v''$.
Note that $s'$ is the leftmost vertex of $E_{v', v''}$ and $t'$ is the
rightmost vertex of the zone $O_{v',v}$ (when $v'$ is of type~1)
or of the rightmost zone $E_{v,w}$ (when $v'$ is of type~2, and $w$ is the other ancestor of~$v'$).

Let us see which vertex can be a neighbor of $v'$ in $H\sm usPtv$. We
already know it cannot be an ancestor of $v'$ or be in
$E_{u, v'}\sm s$.  Suppose it is~$v''$. Then, $H$ must contain two
edges incident to $v''$, and none of them can be an internal edge by
Claim~\ref{cl:intEdge}.  Note that $s'v''$ must be an edge of $H$, for
otherwise, the two only available edges are $v''v'''$ and $v''s''$
(where $v'''$ is the right neighbor of $v''$ in $P_{l-1}$ and $s''$ is
the rightmost neighbor of $v''$ in $P_l$), and this yields a contradiction
because $v'''s''\in E(G)$.  Since $s'v''\in E(H)$, $H$ goes through
the path $R = usPtvt'P_ls'v''$. This path has even length, and contains
all vertices of $N_H(v')$.  So, we may replace $R$ by $uv'v''$ in $H$,
to obtain an even hole that contradicts the minimality of $H$. Now we
know that $v''\notin V(H)$.

Since $v'$ has a neighbor in $H\sm usPtv$, and since this neighbor is
not an ancestor of $v'$, is not $v''$, and is not in $E_{u, v'}$, it must be
in $\myBox_{v'}\sm ( V(P) \cup E_{u, v'} )$.  
By Claim~\ref{cl:intEdge}, the only way that $H$ can contain some vertex of 
$\myBox_{v'} \sm ( V(P) \cup E_{u, v'})$ is if $H$ goes through the edge $vt'$, 
in particular through the right escape of $v$ in $\myBox_{v'}$.
Let $t''$ be the rightmost vertex of $E_{v', v''}$. 
Hence, $H$ must go through the path $S = usPtvt'P_lt''$ (see Figure~\ref{fig:proof-Th3.10}).
This path has even length, and contains
all vertices of $N_H(v')$.  So, we may replace $S$ by $uv't''$ in $H$,
to obtain an even hole that contradicts the minimality of $H$. 

\begin{figure}[ht]
	\centering \includegraphics[width=10cm]{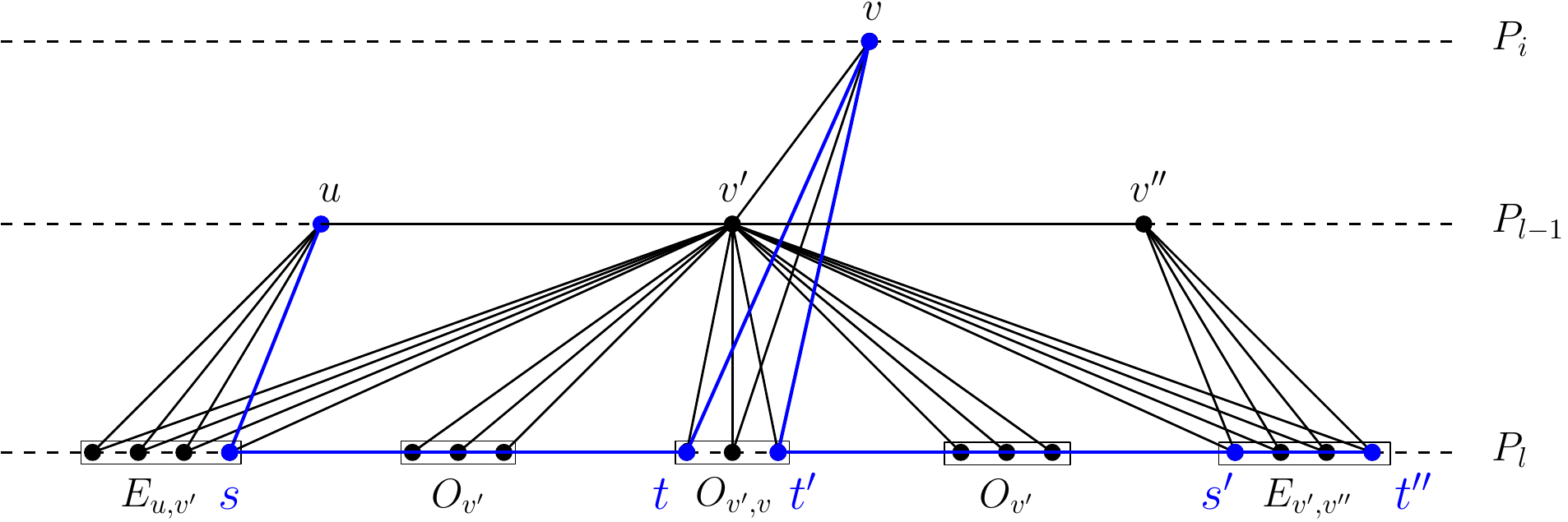}
	\vspace{1cm}
	
	\centering \includegraphics[width=10cm]{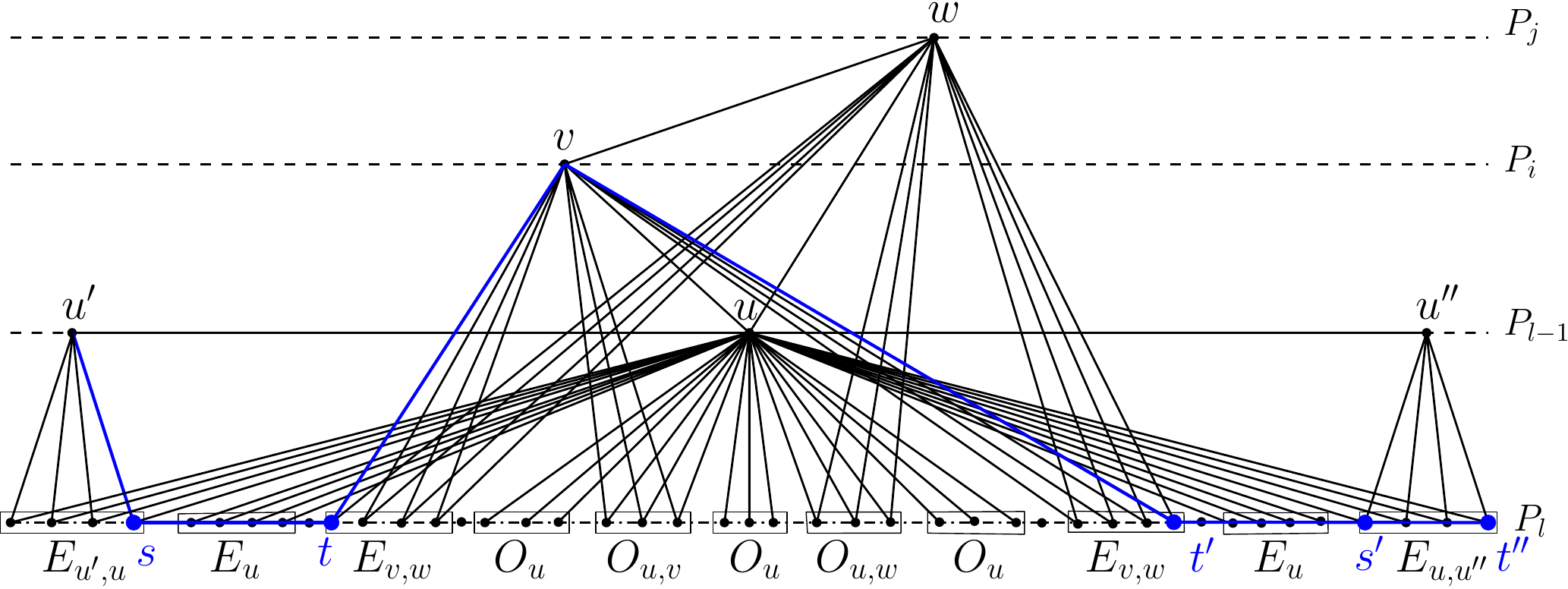}
	\caption{The proof of Theorem~\ref{th:layeredWheel-EHF}: 
	in blue is the path $S = usPtvt'P_lt''$, when $v'$ is of type~1 (top) and when $v'$ is of type~2 (bottom)}
	\label{fig:proof-Th3.10}
\end{figure}
\end{proof}

\medskip


Let us now prove that every ehf-layered-wheel is pyramid-free.

\begin{theorem}
\label{th:layeredWheel-EHF_noPyramid}
For every integer $l \geq 1$, $k \geq 4$, every
$(l, k)$-ehf-layered-wheel $G_{l,k}$ is pyramid-free.
\end{theorem}

\begin{proof}
  \setcounter{claim}{0}
  
  Recall that all layers are viewed as oriented from left to right.
  For a contradiction, suppose that an ehf-layered-wheel $G_{l, k}$
  contains a pyramid $\Pi = 3PC(\tri,x)$. (Here we denote by 
  $\Delta$ the triangle of $\Pi$, and call the apex the only vertex
  of degree~3 in $\Pi \sm \Delta$ which in this case is the vertex~$x$.)
  Suppose that $l$ is minimal,
  and under this assumption that $\Pi$ contains the minimum number of
  vertices among all pyramids in $G_{l,k}$.  Clearly $l \geq 3$, and
  layer $P_l$ contains some vertex of $\Pi$, for otherwise
  $G_{l, k}[P_0 \cup \cdots \cup P_{l-1}]$ would be a counterexample.
  
  The next claim is trivially correct, so we omit the proof.
  
  \begin{claim}\label{cl:hole-inPyramid}
    Any hole in $\Pi$ contains the apex and two vertices of $\tri$.
  \end{claim}

  \begin{claim}{\label{cl:DeltaInZone}}
    If a vertex of $\tri$ is in $P_l$, then it is not in the
    interior of some zone.
  \end{claim}

  \bpc 
  Suppose that some vertex $a$ of $\tri$ is in $P_l$ and is in
  the interior of some zone $Z$. Then $a$ is of type~2. If $Z= E_{u', u}$ for some
  $u', u\in P_{l-1}$, then $\tri = auu'$, and we see that the left
  or the right neighbor of $a$ in $P_l$ is in $\Pi$. Let
  $Q = a \dots b$ be the subpath of $P_l$ that contains $a$, that is
  included in $\Pi$, and that is maximal with respect to these
  properties.  We see that $b$ is adjacent to $u$ and $u'$, so that
  $\Pi$ contains a diamond, a contradiction.  The proof is the same
  for all other kinds of zones (namely $E_{u, u''}$, $E_{v,w}$, $O_{u, v}$,
  or $O_{u, w}$). 
   \epc

  \begin{claim}{\label{cl:apexInPl}}
    The apex $x$ is not in $P_l$.
  \end{claim}

  \bpc Let us see that $x\in P_l$ yields a contradiction.  Since $x$
  has degree~3 in $\Pi$, it is a vertex of type~1 or~2, so it belongs
  to some zone.

  Suppose first that $x$ is in the interior of some zone $Z$.  If
  $Z = E_{u, u'}$ for some $u, u'\in P_{l-1}$, then since $x$ has
  degree~3 in $\Pi$ and is not in a triangle of $\Pi$, we see that the
  two neighbors of $x$ in $P_l$ are in $\Pi$. Also, exactly one
  ancestor $y$ of $x$ must be in $\Pi$.  Let $Q$ be the subpath of
  $P_l$ that contains $x$, that is included in $\Pi$, and that is
  maximal with respect to these properties.  We see that the ends of
  $Q$ are adjacent to $y$, so that $Q$ and $y$ form a cycle with a
  unique chord in a pyramid, while not containing a triangle, a
  contradiction. When $Z$ is another zone, say $E_{u, u''}$, $O_u$,
  $O_{u, v}$, etc, the proof is exactly the same.

  Suppose now that $x$ is an end of some zone $Z$. Again, the two
  neighbors of $x$ in $P_l$ and an ancestor $u$ of $x$ are in $\Pi$.
  So, $\Pi$ contains the path $Q$ from $x$ to the vertex $y$ with
  ancestor $u$ that is next to $x$ along $Z$.  Note that $y$ is in the
  interior of $Z$.  So, $Q$ and $u$ form a hole of $\Pi$.  Apart from
  $x$, $y$, and $u$, every vertex of $H$ has degree~2, so $uy$ is an
  edge of $\tri$, a contradiction to Claim~\ref{cl:DeltaInZone}.  \epc

  \begin{claim}
    \label{cl:notInBox}
    If $u\in P_{l-1}$ has type~0 or~1 and is in $\Pi$, then no internal
    vertex of $\myBox_u$ is in $\Pi$.
  \end{claim}

  \bpc Suppose $a\in \Pi$ is an internal vertex of $\myBox_u$.  Let
  $Q$ be the subpath of $P_l$ that contains $a$, is included in $\Pi$,
  and maximal with respect to this property. Since $a$ is an internal
  vertex of $\myBox_u$ and $u$ has type~0 or~1, $Q$ and $u$ form a
  hole $H$, that must contain the apex. Since by
  Claim~\ref{cl:apexInPl}, the apex is not in $P_l$, it must be $u$,
  and since every internal vertex of $Q$ has degree 2, the two
  neighbors of $u$ in $H$ are in $\tri$, a contradiction since they
  are non-adjacent.
 
  \epc


  
  \begin{claim}{\label{cl:triPl}}
    No vertex of $\tri$ is in $P_l$.
  \end{claim}

  \bpc Suppose for a contradiction that $a$ is a vertex of $\tri$ in
  $P_l$.  So $a$ has type~2, and in particular, it is not an end of
  $P_l$. As every internal vertex of $P_l$, $a$ is in the
  interior of some box $\myBox_u$.   If $u$ is of type~0 or~1, it must
  be part of $\tri$, so $a$ contradicts Claim~\ref{cl:notInBox}.  Hence,
  $u$ is of type~2.  

  We denote by $P=a\dots p$ a subpath of $\Pi$ included in $\myBox_u$
  and maximal with this property. We will now analyze every possible 
  zone where $a$ belongs to, and we will see that each of the cases
  yields a contradiction.
  
  Suppose first that $a$ is in a shared zone $Z$. If $Z=E_{u', u}$,
  and therefore, $\tri = au'u$, then by Claim~\ref{cl:DeltaInZone},
  $a$ is the rightmost vertex of $E_{u', u}$ (since it is in the
  interior of $\myBox_u$).  Since $u'$ is of type 0 (because $u$ is of
  type 2), Claim~\ref{cl:notInBox} applied to $\myBox_{u'}$ implies
  that $a$ is the only vertex of $\Pi$ in $E_{u', u}$, so $p$
  must be the leftmost vertex of the zone $E_u$ that is next to
  $E_{u', u}$.  So $P$ and $u$ form a hole $H$ of $\Pi$, and since $u$
  is in $\tri$, the apex $x$ must be in $P_l$, a contradiction to
  Claim~\ref{cl:apexInPl}. The proof is similar when $Z= E_{u, u''}$. 

  If $Z=O_{u, v}$, then $\tri=auv$, and by Claim~\ref{cl:DeltaInZone},
  $a$ is either the leftmost or the rightmost vertex of $O_{u, v}$.  If $a$
  is the leftmost vertex of $O_{u, v}$, then $p$ is either the
  rightmost vertex of the zone $O_u$ (that is on the left of $O_{u,v}$) 
  or $p$ is the vertex of type~2 next to $a$
  along $O_{u, v}$. In either case, $P$ and $u$ form a hole $H$ of $\Pi$, and since
  $u$ is in the triangle of $\Pi$, the apex $x$ must be in $P_l$, a contradiction
  to Claim~\ref{cl:apexInPl}.  If $a$ is the rightmost vertex of
  $O_{u, v}$, the proof is similar.
  By symmetry, the case when $Z= O_{u, w}$ yields the same contradiction.
  
  When $a$ is the rightmost vertex of the leftmost zone $E_{v, w}$ (that is between
  $E_u$ and $O_u$ when oriented from left), we have $\Delta = avw$
  and so $u \notin \Pi$. The proof is again
  the same, with a hole $H$ that goes through $v$. The case when $a$ 
  is the leftmost vertex of the rightmost zone $E_{v,w}$ (that is between
  $O_u$ and $E_u$ when oriented from left) can be done in the similar way.

  We are left with the case when $a$ is the leftmost vertex of the
  leftmost zone $E_{v, w}$, or the rightmost vertex of the rightmost zone $E_{v, w}$.
  These two cases are symmetric, so we may assume that
  $a$ is the leftmost vertex of the leftmost zone $E_{v, w}$.
  
  It then follows that $\tri= avw$.  
  Note that $u \notin \Pi$ because a pyramid has only
  one triangle.  If $P$ goes in the interior of the zone $E_{v, w}$,
  then $\Pi$ contains a diamond, a contradiction. So, $P$ goes through
  the zone $E_{u}$ that is left to $E_{v, w}$ and contains the rightmost 
  vertex of $E_{u', u}$. Furthermore, there are two cases: $P$ contains
  the zone $E_{u', u}$ (so $p$ is the leftmost vertex of $E_{u', u}$
  and $u'\notin \Pi$), or $P$ contains only the rightmost vertex or
  $E_{u', u}$ (so $p$ is the rightmost vertex of $E_{u', u}$ and
  $u'\in \Pi$).  In the first case, we remove $P \sm p$ from $\Pi$ and put instead
  the edge $up$; in the second case, we remove $P$ from $\Pi$ completely
  and put the edge $uu'$.
    We obtain a pyramid (with triangle $uvw$)
  that is of smaller size than $\Pi$ --- a contradiction, unless $u$ has some
  neighbor in $\Pi \sm (P \cup \{u',v,w\})$. 
  Hence, we now suppose such a neighbor $z$ exists. 

  Let $q$ be the leftmost vertex of the leftmost zone $O_u$ (that is first met when traversing the 
  layer from left to right), and $r$ be the
  rightmost vertex of $E_{u, u''}$.  Consider the path $Q = q P_l r$. 
  Observe that $z$ is in $Q\cup \{u''\}$ because $Q\cup \{u''\}$
  contains all possible neighbors of $u$ in $\Pi\sm (P \cup \{u',v,w\})$.

  Suppose that some vertex of $Q$ is in $\Pi$.  Let $z'$ be the vertex
  of $\Pi$ in $Q$ that is the closest to $q$ along $Q$.  Note that by
  Claim~\ref{cl:apexInPl}, $z'$ has degree~2 in $\Pi$.  Since $z'$ is the
  closest vertex to $q$, it has a neighbor in $\Pi \sm Q$. In
  particular, $z'$ is a type~1 or type~2 vertex, and exactly one of its
  ancestor is in $\Pi$.  Since $u \notin \Pi$, such an ancestor is $v$
  or $w$, or possibly $u''$ if $u''\in \Pi$ (and only one of them).
  If $z' \in O_{u,v}$, or $z' \in O_{u,w}$, or $z' \in E_{v,w}$, 
  then there exists a vertex
  $z'' \in Q$ such that $vz'P_lz''v$, or $wz'P_lz''w$, or $vz'P_lz''wv$
  is a hole of $\Pi$, which in either case contradicts
  Claim~\ref{cl:apexInPl}.  So $z' \in E_{u,u''}$ and the ancestor of
  $z'$ in $\Pi$ must be $u''$ (in particular $u''\in \Pi$). But then,
  the right neighbor of $z'$ in $P_l$ is an internal vertex of
  $\myBox_{u''}$ that belongs to $\Pi$, a contradiction to
  Claim~\ref{cl:notInBox}.  Therefore, $\Pi \cap Q = \emptyset$.
  
  This means that $z = u''$.  Note that the neighbors of $u''$ in $\Pi$
  cannot contain $u$ (because $u\notin\Pi$), cannot be in $E_{u, u''}$
  (because $E_{u, u''}$ is subpath of $Q$), cannot be in the interior
  of $\myBox_{u''}$ (because $u''$ has type~0 and by
  Claim~\ref{cl:notInBox}), so they are precisely the right neighbor
  $u'''$ of $u''$ in $P_{l-1}$ and the rightmost vertex $b$ of
  $E_{u'', u'''}$.  But then, $u''u'''b$ is a triangle in $\Pi$, a
  contradiction. \epc

%



  \vspace{1ex} 
  
  The rest of the proof is quite similar to the proof of
  Theorem~\ref{th:layeredWheel-EHF}, that ehf-layered-wheel contains no even hole.

  Let $P = s \dots t$ be a subpath of $\Pi$ in $P_l$ such that $P$ is
  inclusion-wise maximal (and $s,t$ appear in this order from left to
  right).  By Claims~\ref{cl:apexInPl} and~\ref{cl:triPl}, every
  vertex of $P$ has degree~2 in $\Pi$.  Moreover by the maximality of
  $P$, each of $s$ and $t$ has an ancestor which is also in $\Pi$.
  Note that $s \neq t$, for otherwise $s$ would be of type~2, and
  together with its ancestors, it forms a triangle, which contradicts
  Claim~\ref{cl:triPl}.  Let $u$ and $v$ be the ancestors
  of $s$ and $t$ respectively, such that $u,v \in V(\Pi)$.  By
  Claims~\ref{cl:hole-inPyramid},~\ref{cl:apexInPl},
  and~\ref{cl:triPl}, $u \neq v$ and $uv \notin E(G)$.

  \vspace{1 ex}
  
  \begin{claim}
    \label{cl:not-all-in-P-pyramid}
    For every vertex $p \in V(P_{l-1})$, $N(p) \cap V({P_l}) \not\seq V(P)$.
  \end{claim}
  
  \bpc 
    Suppose that $p \in V(P_{l-1})$ and
    $N(p) \cap V({P_l}) \seq V(P)$.  So, $p\notin V(\Pi)$.  Note that $p$
    is an internal vertex of $P_{l-1}$, for otherwise, $s$ or $t$ is an
    end of $P_l$ and has degree~2, while having two neighbors in
    $V(\Pi) \cap V(G_{l, k}\sm p)$, a contradiction.

    By~\ref{axi:B5}, ancestors of $p$ (if any) and the neighbors of $p$
    in $P_{l-1}$ must also have neighbors in $P$. Thus, all of such
    vertices do not belong to $\Pi$ because $P$ is a subpath of~$\Pi$. Hence,
    replacing $\myBox_p = p' \dots p''$ in $\Pi$ with $p'pp''$ yields a 
    pyramid with strictly less vertices than $\Pi$, a contradiction to the
    minimality of~$\Pi$. 
  \epc

\vspace{1ex}
  
Let $a$ be a vertex in $P_{i}$ for some $0\leq i<l$, and $p$ be a
neighbor of $a$ in $P_{l}$.  We say that $ap$ is an \emph{internal
  edge} if one of the following holds:
  
  \begin{itemize}
    \item $i=l-1$ and $p$ is an internal vertex of $\myBox_a$.
    \item $i<l-1$, $a$ is an ancestor of some $a'\in V(P_{l-1})$, $a'$ has
    type~1 or~2, $p$ is in $\myBox_{a'}$ and $p$ is neither the
    leftmost neighbor of $a$ in $\myBox_{a'}$ nor the rightmost
    neighbor of $a$ in $\myBox_{a'}$.
  \end{itemize}
  
  \begin{claim}{\label{cl:intEdge-pyramid}}
    No internal edge is an edge of $\Pi$. 
  \end{claim}
  
  \bpc Suppose that $p\in P_l$ is the end of an internal edge $e$ that
  is also an edge of $\Pi$.  If the other end of $e$ is in $P_{l-1}$,
  we set $e=pu$ and observe that $p$ is in the interior of $\myBox_u$.
  Otherwise, the other end of $e$ is in $P_i$, with $i<l-1$, we set
  $e=px$ and observe that $x$ has a neighbor $u$ in $P_{l-1}$. Again,
  $p$ is an internal vertex of $\myBox_u$.  Observe that $x$ is either
  $v$ or $w$ as represented on
  Figure~\ref{fig:ehf-layered-wheel}.

  By Claims~\ref{cl:apexInPl}
  and~\ref{cl:triPl}, $p$ has degree~2 in $\Pi$, so $p$ has a unique
  neighbor in $\Pi \cap P_l$.    Let $P = p\dots p'$ be the subpath of
  $P_l$ included in $\Pi$, containing $p$, and maximal with
  respect to this property. 

  It can be checked in Figure~\ref{fig:ehf-layered-wheel} that
  $P$ together with $u$, $v$, $w$, $uv$, $uw$, or~$vw$ forms a hole,
  that contains the apex and two vertices of $\tri$ (by
  Claim~\ref{cl:hole-inPyramid}), a contradiction to
  Claims~\ref{cl:apexInPl} and~\ref{cl:triPl}.  \epc

\vspace{1 ex}
  
  \begin{claim}
    \label{cl:where-is-u,vPyr}
    Exactly one of $u$ and $v$ is in $P_{l-1}$. 
  \end{claim}
  
  \bpc Suppose that both $u$ and $v$ are not in $P_{l-1}$. Since $u$
  and $v$ have neighbors in $P$, each of them has a neighbor $u'$ and
  $v'$ respectively in $P_{l-1}$, such that $s \in \myBox_{u'}$ and
  $t \in \myBox_{v'}$.  

  If $u' = v'$, then $u'$ is a type~2 vertex in~$P_{l-1}$, with ancestors 
  $u$ and $v$, hence $u$ and $v$ are adjacent.  
  It then follows that
  $usPtvu$ is a hole of~$\Sigma$, so it must contains the apex 
  and two vertices of~$\Delta$, contradicting Claim~\ref{cl:apexInPl} or
  Claim~\ref{cl:triPl}, since $u$ and $v$ are the only vertices of the hole
  that are not in~$P_l$.  
  
  Since $u'$ and $v'$ are vertices with
  ancestors, by construction, the interior of $u'P_{l-1}v'$ contains a
  vertex $w$ of type~0. It yields that $N_{P_l}(w)$ is all contained
  in $P$, a contradiction to Claim~\ref{cl:not-all-in-P-pyramid}.

  Suppose now that both $u$ and $v$ are in $P_{l-1}$.  By
  Claim~\ref{cl:not-all-in-P-pyramid}, no vertex of $P_{l-1}$ has all its
  neighbors in $P$.  So the interior of $uP_{l-1}v$ contains at most
  two vertices.

  If $u=v$, then $usPtu$ is a hole of $\Pi$. Since $u$ is 
  the only vertex in the hole that is not in $P_l$,
  by Claim~\ref{cl:hole-inPyramid}
  $P$ contains the apex or a vertex of $\tri$, a contradiction to
  Claims~\ref{cl:apexInPl} or~\ref{cl:triPl}.
  Similarly if $uv\in E(G)$, then $usPtvu$ is a hole of $\Pi$,
  this again yields a contradiction.

  If the interior of $uP_{l-1}v$ contains a single vertex, then let
  $w$ be this vertex.  Let $w_1$ (resp.\ $w_2$) be the neighbor of
  $w$ in $P$ that is closest to $s$ (resp.\ $t$).  It follows by
  construction, that $s=w_1, t=w_2$ (because both $u$ and $v$ are
  adjacent to $w$ in $P_{l-1}$). By Claim~\ref{cl:intEdge-pyramid},
  $\{s\} = V(E_{u, w}) \cap V(\Pi)$ and
  $\{t\} = V(E_{w, v}) \cap V(\Pi)$.  Also, if $w$ has an ancestor,
  then such an ancestor must have neighbors in $P$, and hence it does
  not belong to $\Pi$.  Altogether, we see that
  $N_\Pi(w)\subseteq V(usPtv)$.  So, replacing $usPtv$ in $\Pi$ with
  $uwv$ returns a pyramid with less vertices than $\Pi$, a
  contradiction to the minimality of $\Pi$.

  So the interior of $uP_{l-1}v$ contains two vertices.  We let
  $uP_{l-1}v = uww'v$, and $w_1$ (resp.\ $w'_2$) be the neighbor of
  $w$ (resp.\ $w'$) in $P$ that is closest to $s$ (resp.~$t$).
  Similar as in the previous case, we know that $s=w_1, t=w'_2$; and by
  Claim~\ref{cl:intEdge-pyramid}, $\{s\} = V(E_{u, w}) \cap V(\Pi)$
  and $\{t\} = V(E_{w', v}) \cap V(\Pi)$.  Also, if $w$ or $w'$ has an
  ancestor, then such an ancestor must have neighbors in $P$, and
  hence it does not belong to $\Pi$.  Altogether, we see that
  $N_\Pi(\{w, w'\})\subseteq V(usPtv)$.  So, replacing $usPtv$ in
  $\Pi$ with $uww'v$ returns a pyramid with less vertices than $\Pi$,
  again a contradiction to the minimality of $\Pi$.  \epc

\vspace{1 ex}

By Claim~\ref{cl:where-is-u,vPyr} and up to symmetry, we may assume that
$u \in V(P_{l-1})$ and $v \notin V(P_{l-1})$.  So, $v$ has a neighbor
$v'$ in $P_{l-1}$ such that $t\in \myBox_{v'}$. Note that $v' \neq u$,
for otherwise $usPtvu$ is a hole of~$\Sigma$, so it contains the apex
and two vertices of~$\Delta$, a contradiction to Claim~\ref{cl:apexInPl}
or Claim~\ref{cl:triPl} (because $u$ and $v$ are the only vertices of the
hole that are not in~$P_l$). 
Furthermore, note that the path
$uP_{l-1}v'$ has length at most two, for otherwise some vertex in the
interior of $uP_{l-1}v'$ contradicts
Claim~\ref{cl:not-all-in-P-pyramid}.

Suppose that $uP_{l-1}v'$ has length two, so $uP_{l-1}v' = uwv'$ for
some vertex $w \in V(P_{l-1})$. Then $w$ is of type~0 because $v'$ is
not of type~0.  Let $w'$ be the rightmost vertex of the shared zone
$E_{w, v'}$. By Claim~\ref{cl:intEdge-pyramid},
$\{s\} = V(E_{u, w}) \cap V(\Pi)$ and since $w$ has type~0, we see
that $N_\Pi(w)\subseteq V(usPw')$.  So, replacing $usPw'$ in $\Pi$
with $uww'$ returns a pyramid with less vertices than $\Pi$, a
contradiction to the minimality of $\Pi$.
  
Hence, $uP_{l-1}v'$ has length one, i.e.\ $uP_{l-1}v' = uv'$.  By
Claim~\ref{cl:intEdge-pyramid}, $\{s\} = V(E_{u, v'}) \cap V(\Pi)$.
Observe that $P$ is the left escape of $v$ in $\myBox_{v'}$.  So, $P$
goes through the zone $O_{v'}$ (when $v'$ has type~1) or through the
zone $E_{v'}$ (when $v'$ has type~2).  In particular $v'\notin \Pi$.

If $N_\Pi(v') \subseteq V(usPtv)$, then replacing $usPtv$ in $\Pi$
with $uv'v$ returns a pyramid with less vertices than $\Pi$, a
contradiction to the minimality of $\Pi$.  So, $v'$ has neighbors in
$\Pi$ that are not in $usPtv$.  Note that if $v'$ is of type~2, the
ancestor of $v'$ that is different from $v$ is not in $\Pi$ (because
it is adjacent to $t$ and to $v$, but $t \notin \tri$ by
Claim~\ref{cl:triPl}).

We denote by $v''$ the right neighbor of $v'$ in $P_{l-1}$. Note that
$v''$ has type~0, since $v'$ has type~1 or~2.  Let $s'$ and $t'$ be
vertices such that $t'P_ls'$ is the right escape of $v$ in
$\myBox_{v'}$, $t'$ is adjacent to $v$ and $s'$ is adjacent to $v''$.
Note that $s'$ is the leftmost vertex of $E_{v', v''}$ and $t'$ is the
rightmost vertex of the zone $O_{v',v}$ (when $v'$ is of type~1)
or of the rightmost zone $E_{v,w}$ (when $v'$ is of type~2, and $w$ is the other ancestor of~$v'$).

Let us see which vertex can be a neighbor of $v'$ in $\Pi \sm
usPtv$. We already know it cannot be an ancestor of $v'$ or be a
vertex of $E_{u, v'}\sm s$.  Suppose it is $v''$. Then, $\Pi$ must
contain two edges incident to $v''$, and none of them can be an
internal edge by Claim~\ref{cl:intEdge-pyramid}.  Note that $s'v''$
must be an edge of $\Pi$, for otherwise, the two only available edges are
$v''v'''$ and $v''s''$ (where $v'''$ is the right neighbor of $v''$ in
$P_{l-1}$ and $s''$ is the rightmost neighbor of $v''$ in $P_l$), and
this is a contradiction because $v'''s''\in E(G)$.  Since
$s'v''\in E(\Pi)$, $\Pi$ goes through the path
$R=usPtvt'P_ls'v''$. This path contains all vertices of $N_\Pi(v')$.
Note that $v\notin\tri$, because if so, one of $t$ or $t'$ should be
in $\tri$, a contradiction to Claim~\ref{cl:triPl}.  But $v$ can be
the apex.  If $v$ is not the apex, we may replace $R$ by $uv'v''$ in
$\Pi$, to obtain a pyramid that contradicts the minimality of $\Pi$.
If $v$ is the apex, then we may replace $R\sm v$ by $uv'v''$ in $\Pi$,
to obtain a pyramid with apex $v'$ that contradicts the minimality of
$\Pi$.  Now we know that $v''\notin V(\Pi)$.

Since $v'$ has a neighbor in $\Pi \sm usPtv$, and since this neighbor
is not an ancestor of $v'$, is not $v''$, and is not in $E_{u, v'}$, it must
be in $\myBox_{v'} \sm ( V(P) \cup E_{u, v'})$.  By
Claim~\ref{cl:intEdge-pyramid}, the only way that $\Pi$ can contain
some vertex of $\myBox_{v'} \sm ( V(P) \cup E_{u, v'})$ is that $\Pi$
goes through the edge $vt'$, in particular through the right escape of
$v$ in $\myBox_{v'}$ and through the zone $E_{v', v''}$.  Let $t''$ be
the rightmost vertex of $E_{v', v''}$.  Hence, $\Pi$ must go through
the path $S = usPtvt'P_lt''$.  This path contains all vertices of
$N_\Pi(v')$. Note that $v\notin\tri$, because if so, one of $t$ or $t'$ should be
in $\tri$, a contradiction to Claim~\ref{cl:triPl}.  But $v$ can be
the apex. If $v$ is not the apex, we may replace $S$ by $uv't''$
in $\Pi$, to obtain a pyramid that contradicts the minimality of
$\Pi$.  If $v$ is the apex, then we may replace $S\sm v$ by $uv't''$
in $\Pi$, to obtain a pyramid with apex $v'$ that contradicts the
minimality of $\Pi$.
\end{proof}

\subsection*{Treewidth and cliquewidth}
\label{sec:lower-bound}

For any $l \geq 0$, ttf-layered-wheels and ehf-layered-wheels on $l+1$
layers contain $K_{l+1}$ as a minor. To see this, note that each vertex in layer
$P_i$, $i<l$, has neighbors in all layers $i+1,\dots, l$ (see
Lemma~\ref{lem:neighborhood-ttf-layered-wheel} and
Lemma~\ref{lem:neighborhood-ehf-layered-wheel}).  Hence, by contracting
each layer into a single vertex, a complete graph on $l+1$ vertices is
obtained.  Since when $H$ is a minor of $G$ we have
$\tw(H) \leq \tw(G)$ and since for $l \geq 1$, a complete graph on $l$
vertices has treewidth $l-1$, we obtain the following.

\begin{theorem} \label{th:largeTW}
  For any $l \geq 0$, ttf-layered-wheels and ehf-layered-wheels on
  $l+1$ layers have treewidth at least~$l$.
\end{theorem}

\nocite{DBLP:conf/wg/2000}

Gurski and Wanke~\cite{GurskiW00} proved that the
treewidth is in some sense equivalent to the cliquewidth when some
complete bipartite graph is excluded as a subgraph.  Let us state and
apply this formally (thanks to Sang-il Oum for pointing this out to us).

\begin{theorem}[Gurski and Wanke~\cite{GurskiW00}] \label{th:gurski}
  If a graph $G$ contains no $K_{3, 3}$ as a subgraph, then
  $\tw(G) \leq 6\cw(G)-1$.
\end{theorem}
  
\begin{lemma} \label{lem:noK33}
  A layered wheel (ttf or ehf) contains no $K_{3, 3}$ as a subgraph.
\end{lemma}

\begin{proof}
  Suppose that a ttf-layered-wheel $G$ contains $K_{3, 3}$ as a subgraph. Then,
  either it contains a theta (if $K_{3, 3}$ is an induced subgraph of
  $G$) or it contains a triangle (if $K_{3, 3}$ is not an induced subgraph of
  $G$). In both cases, there is contradiction.

  Suppose that an ehf-layered-wheel $G$ contains $K_{3, 3}$ as a
  subgraph. If one side of the $K_{3, 3}$ is a clique, then $G$
  contains a $K_4$.  Otherwise, each side of $K_{3, 3}$ contains a
  non-edge, so $G$ contains $K_{2,2}$, that is isomorphic to a $C_4$.
  In both cases, there is contradiction.
\end{proof}

\begin{theorem} \label{th:largeCW}
  For any integers $l\geq 2$, $k \geq 4$, the cliquewidth of a layered wheel
  $G_{l,k}$ is at least~$\frac{l+1}{6}$.
\end{theorem}

\begin{proof}
  Follows from Lemma~\ref{lem:noK33} and Theorems~\ref{th:gurski}
  and~\ref{th:largeTW}.  
\end{proof}

\subsection*{Observations and open questions}

It should be pointed out that by carefully subdividing, one may obtain
bipartite ttf-layered-wheels on any number $l$ of layers.  This is
easy to prove by induction on $l$. We just sketch the main step of the
proof: when building the last layer, assuming that the previous layers
induce a bipartite graph, only the vertices with ancestors are
assigned to one side of the bipartition (and only to one side, since a
vertex has at most one ancestor in a ttf-layered-wheels). The parity
of the paths linking vertices with ancestors can be adjusted to
produce a bipartite graph.

It is easy to see that every prism, every theta, and every even wheel
contains an even hole.  Therefore, by
Theorem~\ref{th:layeredWheel-EHF} and Theorem~\ref{th:layeredWheel-EHF_noPyramid}, 
ehf-layered-wheels are (prism,
pyramid, theta, even wheel)-free, which is not obvious from their definitions.
Note that ehf-layered-wheels contain diamonds (recall Conjecture~\ref{conj:diamond}
we proposed in Section~\ref{sec:introduction}).

However, we note that it is possible to modify Construction~\ref{cons:G2(l,k)} 
in such a way that we obtain a layered wheel that is even-hole-free
but contains a pyramid. Such a construction might be of interest 
to see what amount of structure one can get in a even-hole-free 
graphs by studying how the graph attaches to a pyramid. 
The construction is done by modifying axiom~\ref{axi:B5}
where the two zones $E_u$'s are obliterated. More specifically, 
if $u$ is of type~2 (so it is an internal vertex
    of $P_{l-1}$), then let $v \in P_i$ and $w \in P_{j}$, $j \leq i$
    be its ancestors.  In this case,
    $\myBox_u$ is made of only 9 zones, namely $E_{u', u}$, $E_{v, w}$, $O_u$,
    $O_{u, v}$, $O_{u}$, $O_{u, w}$, $O_u$, $E_{v, w}$, and
    $E_{u, u''}$ (see Figure~\ref{fig:ehf-layered-wheel-withPyramid}). 
The fact that this modified construction keeps the property of the layered wheel
being even-hole-free can be proven similarly as Theorem~\ref{th:layeredWheel-EHF}.
Notice that Lemma~\ref{lem:length-of-boxes} also remains true for this modified construction.
We remark that a corresponding wheel that is even-hole-free 
(similar to the one in Figure~\ref{fig:ehf-LW-wheel})
exists considering this modified pattern of zones.

\begin{figure}[ht]
	\begin{center}
    	\includegraphics[width=10cm]{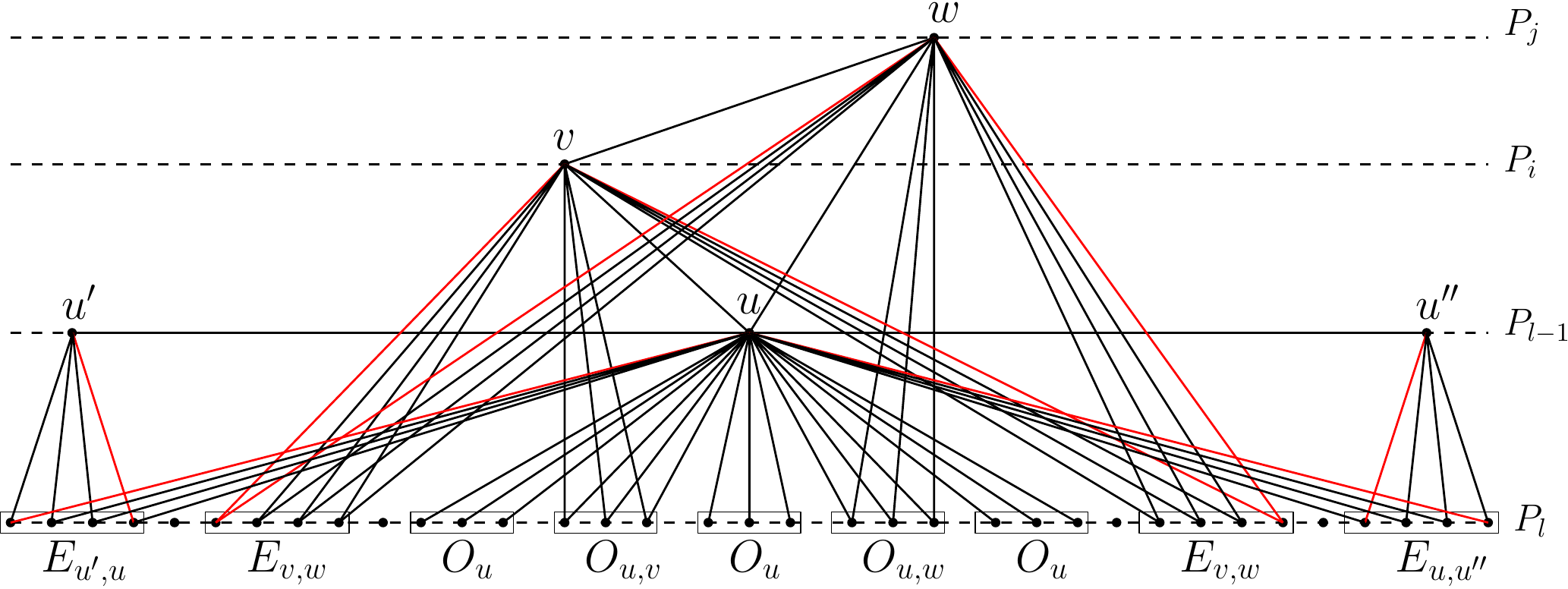}
  	\end{center}
  	\caption{A modified construction of ehf-layered-wheel $G_{l,k}$ which contain pyramids
  	(dashed lines between two vertices in $P_l$ represent paths of odd length)} 
  	\label{fig:ehf-layered-wheel-withPyramid}
\end{figure}

\begin{figure}[ht]
	\begin{center}
    	\includegraphics[width=10cm]{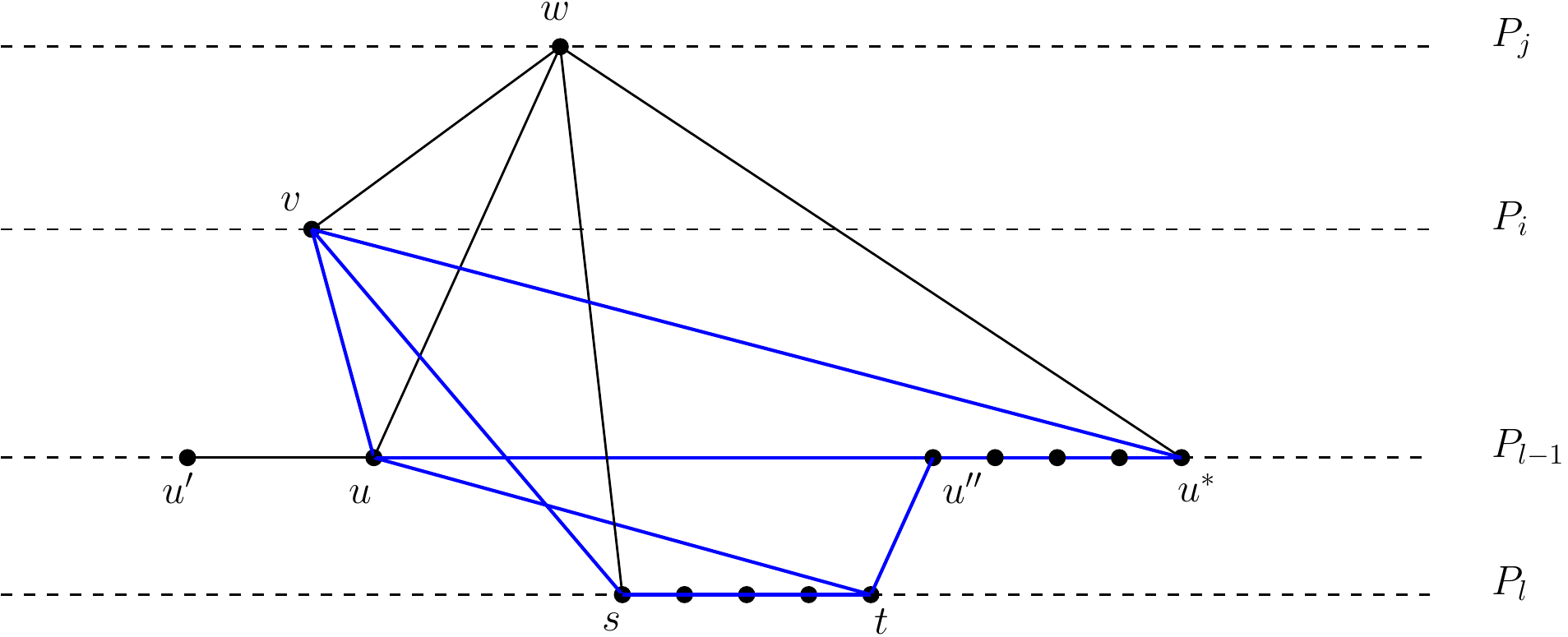}
  	\end{center}
  	\caption{A pyramid (in blue) that is contained in 
  	a modified ehf-layered-wheel $G_{l,k}$, for some integers $l,k$
  	\label{fig:pyamid-LW}}
\end{figure}

An example of a pyramid that may be found in such a modified
ehf-layered-wheel is given in Figure~\ref{fig:pyamid-LW}.  In the
figure, $u \in P_{l-1}$ is a type~2 vertex with ancestors $v \in P_i$
and $w \in P_j$, $j < i$, $u^{*}$ is a common neighbor of $v$ and
$w$ in $P_{l-1}$ such that $u$ and $u^{*}$ are consecutive common
neighbors of $v$ and~$w$ in some $vw$-zone in $P_{l-1}$, $s$ is the
rightmost vertex of a zone labeled $E_{v,w} \subseteq \myBox_u$ in
$P_l$; and $t$ is the leftmost vertex of the zone labeled
$E_{u,u''} \subseteq \myBox_u$ in $P_l$ where $u''\neq u'$ is adjacent to $u$
in $P_{l-1}$. The pyramid has triangle $utu''$ and apex $v$.

\section{Lower bound on rankwidth}
\label{sec:lower-bound-rankwidth}

In this section, we prove that there exist ttf-layered-wheels and
ehf-layered-wheels with arbitrarily large rankwidth.  This follows
directly from Theorem~\ref{th:largeCW} and Lemma~\ref{lem:width-comparation}, but
by a direct computation, we provide a better bound. Let us first
present some useful notion and definition about rankwidth.

For a set $X$, let $2^X$ denote the set of all subsets of $X$. For
sets $R$ and $C$, an $(R,C)$-matrix is a matrix where the rows are
indexed by elements in $R$ and columns are indexed by elements in $C$. For
an $(R,C)$-matrix $M$, if $X \subsetneq R$ and $Y \subsetneq C$, we
let $M[X,Y]$ be the submatrix of $M$ where the rows and the columns
are indexed by $X$ and $Y$ respectively. For a graph $G = (V,E)$, let
$A_G$ denote the adjacency matrix of $G$ over the binary field (i.e.,
$A_G$ is the $(V,V)$-matrix, where an entry is 1 if the column-vertex
is adjacent to the row-vertex, and 0 otherwise). The \emph{cutrank
  function} of $G$ is the function
$\textrm{cutrk}_G : 2^V \rightarrow \mathbb{N}$, given by
$$\textrm{cutrk}_G(X) = \rank(A_G [X, V \sm X]),$$ where the
rank is taken over the binary field.

A tree is a connected, acyclic graph. 
A \emph{leaf} of a tree is a vertex incident to exactly one edge. 
For a tree $T$, we let $L(T)$ denote the set of all leaves of $T$. 
A tree vertex that is not a leaf is called \emph{internal}. 
A tree is \emph{cubic}, if it has at least two vertices and every internal vertex has degree 3.

A \emph{rank decomposition} of a graph $G$ is a pair $(T, \lambda)$,
where $T$ is a cubic tree and $\lambda : V(G) \rightarrow L(T)$ is a
bijection. If $|V(G)| \leq 1$, then $G$ has no rank decomposition. For
every edge $e \in E(T)$, the connected components of $T \sm e$ induce a
partition $(A_e, B_e)$ of $L(T)$. The \emph{width of an edge} $e$ is
defined as $\textrm{cutrk}_G (\lambda^{-1}(A_e))$. The \emph{width of
  $(T, \lambda)$}, denoted by $\textrm{width}(T,\lambda)$, is the
maximum width over all edges of $T$.  The \emph{rankwidth of $G$},
denoted by $\rw(G)$, is the minimum integer $k$, such that there is a
rank decomposition of $G$ of width $k$. (If $|V(G)| \leq 1$, we let
$\rw(G) = 0$.)  The next lemma follows directly from the definition of
the rankwidth.

\begin{lemma}
\label{lem:rankwidth-of-induced-subgraph}
  Let $G$ be a graph and $H$ be an induced subgraph of $G$. Then
  $\rw(H) \leq \rw(G)$.
\end{lemma}

\vspace{1 ex}

A class $\mathcal{C}$ of graphs has \emph{bounded rankwidth} if there
exists a constant $k \in \mathbb{N}$, such that every $G \in \mathcal{C}$
satisfies $\rw(G) \leq k$. If such a constant does not exist, then
$\mathcal{C}$ has \emph{unbounded rankwidth}.  In the following
lemmas, we present some basic properties related to rankwidth.  Let
$T$ be a tree, we call an edge $e \in E(T)$ \emph{balanced}, if the
partition $(A_e, B_e)$ of $L(T)$ satisfies
$\frac{1}{3}|L(T)| \leq |A_e|$ and $\frac{1}{3}|L(T)| \leq |B_e|$. The
following is well-known (we include a proof for the sake of
completeness).


\begin{lemma} 
\label{lem:cubic-tree}
  Every cubic tree has a balanced edge.
\end{lemma}

\begin{proof} 
  \setcounter{claim}{0} Let $T$ be a cubic tree with $n$ leaves.  We may
  assume that $n \geq 3$, for otherwise, $T$ is a path of length~1,
  and the only edge of $T$ is balanced.

  Let $e=ab$ be an edge of $T$ such that the set of leaves $A_e$ of the 
  connected component of
  $T\sm e$ that contains $a$, satisfies $|A_e|\geq |L(T)|/3$. Suppose 
  that $a$ and $b$ are chosen subject to the minimality of $|A_e|$.
  If $|A_e|\leq 2|L(T)|/3$, then $e$ is balanced. Otherwise,
  $|A_e| > 2|L(T)|/3 \geq 2$ so $a$ has two neighbors $a'$, $a''$
  different from $b$.  Let $A'$ (resp.\ $A''$) be the set of leaves of
  the connected component of $T\sm aa'$ (resp.\ $T\sm aa''$) that contains $a'$ (resp.\
  $a''$). Since $|A_e| > 2|L(T)|/3$ and $A_e= A'\cup A''$, either
  $|A'| > |L(T)|/3$ or $|A''| > |L(T)|/3$.  Hence, one of $A'$ or
  $A''$ contradicts the minimality of $|A_e|$.
\end{proof}

%

\vspace{1 ex}

Let us now introduce a notion that is useful to describe how we can 
represent the structure of layered wheels into a matrix. 
An $n\times n$ matrix $M$ is \emph{fuzzy triangular} if $m_{1, 1} = 1$
and for every $i \in \{2,\dots,n\}$, $m_{i,i} = 1$ and either
$m_{1,i} = m_{2,i} = \cdots = m_{i-1,i} = 0$ or
$m_{i,1} = m_{i,2} = \cdots = m_{i,i-1} = 0$.

\begin{lemma} 
\label{lem:rank}
  Every $n \times n$ fuzzy triangular matrix has rank $n$.
\end{lemma}

\begin{proof}
  Let $M$ be an $n\times n$ fuzzy triangular matrix.  We prove by
  induction on~$n$, that $\rank(M) = n$.  For $n = 1$, this trivially
  holds. Suppose that $n\geq 2$.  If
  $m_{1,n} = m_{2,n} = \cdots = m_{n-1,n} = 0$, we show that rows
  $r_1,\dots, r_n$ of $M$ are linearly independent.  Let
  $\lambda_1,\dots, \lambda_n \in [0,1]$ be such that
  $\Sigma_{i=1}^n \lambda_i r_i = \textbf{0}$ (where $\textbf{0}$
  is the zero vector of length $n$).  Since $m_{n,n} = 1$, we have
  $\lambda_n = 0$. This implies that
  $\Sigma_{i=1}^{n-1} \lambda_i r'_i = \textbf{0}$, where $r'_i$ is the
  row obtained from $r_i$ by deleting its last entry.  Since
  $r'_1,\dots, r'_{n-1}$ are the rows of an $(n-1)\times (n-1)$
  fuzzy triangular matrix, they are linearly independent by the
  induction hypothesis, so $\lambda_1 = \cdots = \lambda_{n-1} = 0$.

  We can prove in the same way that, if
  $m_{n,1} = m_{n,2} = \cdots = m_{n,n-1} = 0$, then the set of $n$
  columns of $M$ is linearly independent.  This shows that
  $\rank(M) = n$.
\end{proof}

Let $G$ be a graph and $(X,Y)$ be a partition of $V(G)$. A path $P$ in
$G$ is {\em separated by $(X,Y)$} if $V(P) \cap X$ and $V(P) \cap Y$
are both non-empty. Note that when $P$ is separated by $(X,Y)$, 
there exists a {\em separating edge} $xy$ of $P$ whose end-vertices 
are~$x \in X$ and~$y \in Y$.

\begin{lemma} 
\label{lem:separated-paths}
  Let $(T, \lambda)$ be a rank decomposition of width at most $r$ of a
  layered wheel with layers $P_0,P_1,\dots, P_l$. Let $e$ be an
  edge of $T$, and $(X,Y)$ be the partition of $V(G)$ induced by
  $T\sm e$. Then there are at most $r$ paths among $\{P_0,P_1,\dots, P_l\}$ that
  are separated by $(X,Y)$.
\end{lemma}

\begin{proof}
  Suppose for a contradiction that $P_{i_1},\dots, P_{i_{r+1}}$ are
  layers that are all separated by $(X, Y)$, where $1 \leq i_1< \dots <
  i_{r+1} \leq l$.  For each integer $i_j$, consider a separating edge $x_{i_j}y_{i_j}$
  of $P_{i_j}$ such that $x_{i_j}\in X$ and $y_{i_j}\in Y$. Set
  $S_X = \{x_{i_1},\dots, x_{i_{r+1}}\}$ and $S_Y = \{y_{i_1},\dots,
  y_{i_{r+1}}\}$. 

  Consider $M[S_X,S_Y]$, the adjacency matrix whose rows are indexed
  by $S_X$ and columns are indexed by $S_Y$.  The definition of
  layered wheels (see~\ref{axi:A6} and~\ref{axi:B7}) says that when two vertices in a layer are
  adjacent, at most one of them has ancestors. It follows that
  $M[S_X,S_Y]$ is fuzzy triangular.  By Lemma~\ref{lem:rank},
  $M[S_X,S_Y]$ has rank $r+1$, a contradiction, because
  $$ \textrm{width}(T,\lambda) \geq \textrm{cutrk}_{G}(X) = \rank(M[X,Y]) \geq \rank(M[S_X,S_Y]) = r+1.$$
\end{proof}

We need the following lemma in our proof.

\begin{lemma}[See \cite{adlerLMRTV:rwehf}]
  \label{lem:subpaths}
  Let $G$ be a graph and $(T,\lambda)$ be a rank decomposition of~$G$
  whose width is at most $r$.  Let $P$ be an induced path of $G$ and $(X,Y)$
  be the partition of $V(G)$ induced by $T\sm e$ where $e \in E(T)$. Then
  each of $P[X]$ and $P[Y]$ contains at most $r+1$ connected
  components.
\end{lemma}

\vspace{1 ex}

Now we are ready to describe layered wheels for which we can prove
that the rankwidth is unbounded.  Let us first define some 
terminology that is used along the proof.  Recall Construction~\ref{cons:G1(l,k)} of
ttf-layered-wheels.  Let $u$ and $v$ be two vertices that are
adjacent in a layer $P_i$ for some $i \in \{1,\dots, l-1\}$, and
they appear in this order (from left to right) along $P_i$.  Let $a$ be the
rightmost vertex of $\myBox_u$ and $b$ be the leftmost vertex of
$\myBox_v$ in $P_{i+1}$.  Let $a'$ (resp.\ $b'$) be the neighbor of
$a$ (resp.\ $b$) in $P_{i+1} \sm \myBox_u$ (resp.\
$P_{i+1} \sm \myBox_v$).  The path $a'P_{i+1}b'$ is called the
\emph{$uv$-bridge}.  An edge $pq$ in $a'P_{i+1}b'$ is called \emph{the
middle edge of the bridge} if the length of the paths $a'P_{i+1}p$
and $qP_{i+1}b'$ are equal.

We have a similar definition for ehf-layered-wheel.  For adjacent
vertices $u$ and $v$ in $P_{i+1}$, the $uv$-bridge in $P_{i+1}$ is the
zone labelled $E_{u,v} \subseteq \myBox_u \cap \myBox_v$ (that we
called in the previous section a shared part).  Observe that in both
layered wheels, every internal vertex of some layer yields two
bridges, and each end of a layers yields one bridge.  We say that a layered wheel is
\emph{special} if every bridge in all layers has odd length (and
therefore admits a middle edge).  The following lemmas are a direct
consequence of Construction~\ref{cons:G1(l,k)} and
Construction~\ref{cons:G2(l,k)}.

\begin{lemma}
\label{lem:special-ttf-LW}
  For every integers $l \geq 1$ and $k \geq 4$, there exists a special
  $(l,k)$-ttf-layered-wheel.
\end{lemma}

\begin{proof}
  The result follows because by~\ref{axi:A6} of
  Construction~\ref{cons:G1(l,k)}, the path between $\myBox_u$ 
  and~$\myBox_v$ is of length at least~$k-2$. 
  So for any two adjacent vertices in a layer,
  the $uv$-bridge can have any odd length, at least~$k-4$.
\end{proof}

\begin{lemma}
\label{lem:special-ehf-LW}
  For every integers $l \geq 1$ and $k \geq 4$, any ehf-layered-wheel is special.
\end{lemma}

\begin{proof}
  The result follows from the fact that shared parts have odd length (see
  Lemma~\ref{lem:length-of-boxes}).
\end{proof}

\vspace{1 ex}

Let $G_{l,k}$ be a layered wheel that is special.  Let $uv$ be an edge of some
layer~$P_i$, where $1\leq i < l$, such that $u$ and $v$ appear in
this order (from left to right) along $P_i$.  Then we denote by $r_ul_v$ the middle
edge of the $uv$-bridge (again, $r_u$ and $l_v$ appear in this order
from left to right).

For any vertex $v \in P_i$, $1 \leq i < l$, the \emph{domain of $v$} (or the
\emph{$v$-domain}), denoted by $\Dom(v)$ is defined as follows:
\begin{itemize}

\item if $v \in V(P_0)$, then $\Dom(v) = V(P_1)$;
\item if $v$ is an internal vertex of $P_i$, then
  $\Dom(v) = V(l_vP_{i+1}r_v)$;
  \item if $v$ is the left end of $P_i$, then $\Dom(v) = V(pP_{i+1}r_v)$, 
  where $p$ is the leftmost vertex of $\myBox_v$; and
  \item if $v$ is the right end of $P_i$, then $\Dom(v) =  V(l_vP_{i+1}q)$, 
  where $q$ is the rightmost vertex of $\myBox_v$.
\end{itemize}

Note that for ttf-layered-wheels, $\myBox_v$ is completely contained
in the $v$-domain, which is not the case for ehf-layered-wheels.  We
are now ready to describe the layered wheels that we need.

\begin{definition}
  For some integer $m$, a special layered wheel $G_{l,k}$ is
  \emph{$m$-uniform}, if for every vertex $v \in V(P_i)$,
  $0 \leq i \leq l-1$, $\Dom(v)$ contains exactly $m$ vertices.
\end{definition}

\noindent Observe that by definition, any $m$-uniform layered wheel is special.
\vspace{1 ex}

\begin{lemma}
  \label{lem:existM}
  For every integers $l \geq 1$ and $k \geq 4$ and $M$, there exists an integer
  $m\geq M$ and a ttf-layered-wheel that is $m$-uniform. 
\end{lemma}

\begin{proof}
  We construct an $m$-uniform ttf-layered-wheel $G_{l,k}$ by adjusting the
  length obtained in step~\ref{axi:A6} of Construction~\ref{cons:G1(l,k)}.
\end{proof}

\vspace{1 ex}
\begin{lemma}
\label{lem:existM-2}
  For every integers $l \geq 1$ and $k \geq 4$ and $M$, there exists an
  integer $m\geq M$ and an ehf-layered-wheel that is $m$-uniform.
\end{lemma}

\begin{proof}
  We construct an $m$-uniform ehf-layered-wheel $G_{l,k}$ by adjusting
  the length obtained in step~\ref{axi:B7} of
  Construction~\ref{cons:G2(l,k)}.
\end{proof}


For a vertex $v \in P_i$, $0 \leq i \leq l$ and an integer $0 \leq d \leq l-i$, the
\emph{$v$-domain of depth~$d$}, denoted by $\Dom^d(v)$ is defined as
follows.
\begin{itemize}
  \item $\Dom^0(v) = \{v\}$ and $\Dom^1(v) = \Dom(v)$;
  \item $\Dom^d(v) = \bigcup_{x \in \Dom(v)} \Dom^{d-1}(x)$ for $d \geq 1$.
\end{itemize}

\vspace{1 ex}

\begin{observation}
  For every $v \in P_i$ with $0 \leq i \leq l$, and for any
  $0 \leq d \leq l-i$, we have $\Dom^d(v) \seq V(P_{i+d})$, where the
  equality holds when $i = 0$.
\end{observation}


\vspace{1 ex}

\begin{lemma}
\label{lem:domain-LW}
For every $0 \leq i \leq l$ and $0 \leq d \leq i$,
$V(P_i) = \bigcup_{v \in P_{i-d}} \Dom^d(v)$.  Moreover, for any
distinct $u,v \in V(P_{i-d})$, $\Dom^d(u) \cap \Dom^d(v) = \emptyset$.
\end{lemma}


\begin{proof}
The statement simply follows by induction on $d$. 
\end{proof}

\vspace{1 ex}

\begin{lemma}
\label{lem:length-of-domain}
For some integers $l,k,m$, let $G_{l,k}$ be an $m$-uniform layered
wheel.  For every $0 \leq i \leq l-1$, $v \in P_i$, and
$1 \leq d \leq l-i$, we have $|\Dom^d(v)| = m^{d}$.
\end{lemma}

\begin{proof}
  The statement simply follows from Lemma~\ref{lem:domain-LW} and the $m$-uniformity: for any
  vertex $v$, $|\Dom^1(v)| = m$ and
  $|\Dom^d(v)| = m \cdot |\Dom^{d-1}(v)| $.
\end{proof}

\vspace{1 ex}

\begin{lemma}
  \label{lem:thesize-of-domains}
  For some integers $l,k,m$, let $G_{l,k}$ be an $m$-uniform layered
  wheel.  Denote by $G_{i,k}$, the subgraph induced by the first $i+1$
  layers $P_0,P_1,\dots, P_i$ of~$G_{l,k}$.  Then
  $ |V(G_{i,k})| < \frac{1}{m-1}|V(P_{i+1})| $ for
  $0 \leq i \leq l-1$.
\end{lemma}

\begin{proof}
  Recall that $V(P_i) = \Dom^i(r)$ for every $1 \leq i \leq l$, with
  $r \in V(P_0)$.  So by Lemma~\ref{lem:length-of-domain},
  $|V(P_i)| = m^i$.  Moreover,
  $|V(G_{i,k})| = \Sigma_{d=0}^i {|\Dom^d(r)|} =
  \frac{m^{i+1}-1}{m-1}.$ Hence, the result directly follows.
\end{proof}

\vspace{1 ex}


\begin{lemma}
\label{lem:Q_X,Q_Y}
Let $l \geq 2$, $k \geq 4$, and $m \geq 15$ be integers, and
$(T,\lambda)$ be a rank decomposition of an $m$-uniform layered wheel
$G_{l,k}$ of width at most $r$.  Let $e$ be a balanced edge in $T$,
and $(X,Y)$ be the partition of $V(G_{l,k})$ induced by $e$.  
Then $P_l$ is separated by $(X,Y)$, and each of
$X$ and $Y$ contains an induced subpath of $P_l$, namely $P_X$ and
$P_Y$ where:
  $$ |V(P_X)|,|V(P_Y)| \geq \floor*{\frac{|V(P_l)|}{3.5(r+1)}}. $$
\end{lemma}

\begin{proof}

  Let first prove that $P_l$ is separated by $(X,Y)$.
  By Lemma~\ref{lem:thesize-of-domains}, we know
  that $|V(P_l)| > (m-1)|V(G_{l-1,k})|$ where $G_{l-1,k} = G_{l,k} \sm P_l$.
  Since $m-1 \geq 14$, we have $|V(P_l)| > \frac{14}{15}|V(G_{l,k})|$
  Hence, $P_l$ cannot be fully contained in~$X$, for otherwise 
  $|Y| < \frac{1}{15} |V(G_{l,k})|$ that would contradict the fact
  that $(X,Y)$ is a balanced decomposition.
  By the same reason, $P_l$ is not fully contained in~$Y$.
  This proves the first statement.

  For the second statement, we will only prove the existence of $P_X$ 
  (for $P_Y$, the proof is similar).
  Since $e$ is a balanced edge of $T$, we have $|X| \geq \frac{1}{3} |V(G_{l,k})|$.  
  Clearly, $$|V(P_l) \cap X| \geq \frac{1}{3} |V(G_{l,k})| - |V(G_{l-1,k})| = \frac{1}{3} \left( |V(P_l) - 2|V(G_{l-1,k})| \right) .$$ 
  By Lemma~\ref{lem:subpaths}, $X$ contains at most $r+1$ connected components of $P_l$. Hence:
    \begin{align*}
        |V(P_X)| & \geq \frac{|V(P_l) \cap X|}{r+1} \\
          & > \frac{|V(P_l)| - \frac{2}{m-1}|V(P_{l})|}{3(r+1)} \tag{1} \\
          & = \frac{m-3}{3(m-1)(r+1)} \ |V(P_l)| \\
          & \geq \frac{2}{7(r+1)} \ |V(P_l)| \tag{2}
    \end{align*}
  Inequality (1) is obtained from Lemma~\ref{lem:thesize-of-domains},
  and (2) follows because $m \geq 15$.
\end{proof}

\vspace{1 ex}
The following theorem is the main result of this section.


\begin{theorem}
\label{th:lowerbound-rankwidth}
For $l\geq 2$, $k \geq 4$, there exists an integer $m$ such that the
rankwidth of an $m$-uniform layered wheel $G_{l,k}$ is at least~$l$.
\end{theorem}

\begin{proof}
  Set $M=15$ and consider an integer $m$ as in Lemma~\ref{lem:existM}
  (or Lemma~\ref{lem:existM-2}), and let $G_{l,k}$ be $m$-uniform.
  
  Suppose for a contradiction that $\rw(G_{l,k}) = r$ for some
  integer $r \leq l-1$.  Let $(T,\lambda)$ be a rank decomposition of
  $G_{l,k}$ of width $r$, and $e$ be a balanced edge of $T$ that
  partition $V(G_{l,k})$ into $(X,Y)$.  Let
  $\mathcal{P} = \{P_0,P_1,\dots, P_l\}$ be the set of layers in the layered wheel, 
  and $\mathcal{S}$ be the set of paths in $\mathcal{P}$ that are separated by $(X,Y)$.
  By Lemma~\ref{lem:separated-paths}, $|\mathcal{S}| \leq r$.
  
  Note that $P_0 \notin \mathcal{S}$ because it contains a single vertex. 
  So, $\mathcal{P} \sm \mathcal{S} \neq \emptyset$.  
  Let $P_j \in \mathcal{P} \sm \mathcal{S}$,
  i.e., the vertices of $P_j$ are completely contained either in $X$ or $Y$.  
  Without loss of generality, we may assume that $V(P_j) \seq X$. 
  
  \vspace{1ex}
  
  \begin{claim}
  \label{cl:value-of-j}
      There exists some $j$ such that $1 \leq j < l$.
  \end{claim}
  
  \bpc
  	Note that $l-r \geq 1$, because $r \leq l-1$.
  	So it is enough to prove that such a $j \geq l-r$ exists.
  	We know that $|\mathcal{S}| \leq r \leq l-1$.
  	If every path $P_j \in \mathcal{P} \sm \mathcal{S}$ has index $j < l-r$,
  	then $|\mathcal{P} \sm \mathcal{S}| \leq l-r$.
  	This implies $|\mathcal{S}| \geq (l+1) - (l-r) = r+1$, a contradiction,
  	so the left inequality of the statement holds
    (the bound is tight when $\mathcal{S} = \bigcup_{l-r+1 \leq i \leq l} \{P_i\}$). 
	Furthermore, by Lemma~\ref{lem:Q_X,Q_Y}, $P_l \in \mathcal{S}$,
	so for every $P_j$ that satisfies the left inequality, 
	we know that $j < l$.  
  \epc 
  
  \vspace{1ex}
  
  Now by Lemma~\ref{lem:Q_X,Q_Y}, there exists a
  subpath $P_Y$ of $P_l$, such that $V(P_Y) \seq Y$ and
  $|V(P_Y)| \geq \floor*{\frac{|V(P_l)|}{3.5(r+1)}}$,  with $|V(P_l)| = m^l$
  (because $|V(P_l)| = \Dom^l(r)$ where $r \in P_0$).
 
  Let $P'$ be the set of vertices in $P_j$ such that
  $N(v) \cap V(P_Y) \neq \emptyset$ for every $v \in P'$.
  Note that the order (left to right) of the domains of $V(P_j)$ in layer $P_l$
  appear as the order of $V(P_j)$ in $P_j$, and
  by Lemma~\ref{lem:domain-LW}, for every
  $v \neq v' \in P_j$, we have $\Dom^{l-j}(v) \cap \Dom^{l-j}(v') = \emptyset$.
  So $P'$ induces a path.
  Moreover, for each vertex $v \in P'$, we can fix a vertex
  $y_v \in V(P_Y) \cap \Dom^{l-j}(v)$, such that $vy_v \in E(G)$.
  Thus for any $v \neq v' \in P'$, we have
  $y_v \neq y_{v'}$, and in particular, $vy_v, v'y_{v'} \in E(G)$ and $v'y_v, vy_{v'} \notin E(G)$.
  Let us denote $S_X = V(P')$ and $S_Y = \{y_v \ | \ v \in S_X\}$.
  Observe that there is a bijection between $S_X$ and $S_Y$, so
  $M[S_X,S_Y]$ is the identity matrix of size $|S_X|$.
  
  Furthermore, by Lemmas~\ref{lem:domain-LW} and~\ref{lem:length-of-domain}, we have
  $|S_X| \geq \floor* {\frac{|V(P_Y)|}{|\Dom^{l-j}(v)|}} = \floor*
  {\frac{|V(P_Y)|}{m^{l-j}}}$.
  By Claim~\ref{cl:value-of-j}, Lemma~\ref{lem:Q_X,Q_Y}, and taking $m \geq 4 l^2$, the following holds.
  $$ |S_X| \geq \floor*{\frac{m^{l}}{3.5(r+1)m^{l-j}}}
    \geq \floor*{\frac{m^{j}}{3.5(r+1)}} 
  	\geq \floor*{\frac{m}{3.5l}}  
  	\geq \floor*{\frac{3.5l^2}{3.5l}}
  	\geq l $$
  which yields a contradiction, because
 $$ r \geq \textrm{width}(T,\lambda) \geq \textrm{cutrk}_{G_{l,k}}(X) = \rank(M[X,Y]) \geq \rank(M[S_X,S_Y]) = |S_X| \geq l. $$
\end{proof}

\section{Upper bound}
\label{sec:upper-bound}

Layered wheels have an exponential number of vertices in terms of
the number of layers $l$. In Section~\ref{sec:layeredWheels}, we have
seen that the treewidth of layered wheels is lower-bounded by $l$.  In
this section, we give an upper bound of the treewidth of layered
wheels.  As mentioned in the introduction, we indeed prove a stronger result:
the so-called \emph{pathwidth} of layered wheels is upper-bounded by
some linear function of $l$.  Since layered wheels $G_{l,k}$ contain
an exponential number of vertices in terms of the number of layers,
this implies that $\tw(G_{l,k}) = \bigO \left(\log |V(G_{l,k})| \right)$.
Beforehand, let us state some useful notions.

\subsection*{Pathwidth}

A \emph{path decomposition} of
a graph $G$ is defined similarly as a tree decomposition except that
the underlying tree is required to be a path. Similarly, the width of the
path decomposition is the size of a largest bag minus one, and the
pathwidth is the minimum width of a path decomposition of $G$.
The \emph{pathwidth} of a graph $G$ is denoted by $\pw(G)$.  
As outlined in the introduction, 
path decomposition is a special case of tree decomposition.
We restated the following lemma that was already mentioned in 
Lemma~\ref{lem:width-comparation}.

\begin{lemma}
\label{lem:tw-leq-pw}
	For any graph $G$, $\tw(G) \leq \pw(G)$.
\end{lemma}

Let $P$ be a path, and $P_1,\dots, P_k$ be subpaths of $P$. The
\emph{interval graph} associated to $P_1,\dots, P_k$ is the graph
whose vertex set is $\{P_1,\dots, P_k\}$ with an edge between any
pair of paths sharing at least one vertex.  So, interval graphs are
intersection graphs of a set of subpaths of a path.

\begin{lemma} [See Theorem 7.14 of~\cite{CyganFKLMPPS15}] 
\label{lem:pw-interval-graph}
Let $G$ be a graph, and $I$ be an interval graph that contains $G$ as
a subgraph (possibly not induced).  Then $\pw(G) \leq \omega(I) - 1$,
where $\omega(I)$ is the size of the maximum clique of $I$.
\end{lemma}

Now, for every layered wheel $G_{l,k}$, we describe an interval graph
$I(G_{l,k})$ such that $G_{l,k}$ is a subgraph of $I(G_{l,k})$.  We
define the \emph{scope} of a vertex. This is similar to its domain,
but slightly different (the main difference is that scopes may overlap
while domains do not).  For $v \in V(P_i$), where $0 \leq i \leq l-1$,
the scope of $v$, denoted by $\Scope(v)$, is defined as follows.

\vspace{1ex}

\noindent For a ttf-layered-wheel:
\begin{itemize}
  \item if $v \in P_0$, $\Scope(v) = V(P_1)$;
  \item if $v$ is in the interior of $P_{i}$, then
    $\Scope(v) = V(L) \cup \myBox_v \cup V(R)$, where 
    $L$ is the $uv$-{bridge} and $R$ is the $vw$-{bridge}, 
    $u$ and $w$ are the left and the right neighbors of $v$ in $P_{i}$ respectively;
  \item if $v$ is the left end of $P_i$, then
    $\Scope(v) = \myBox_v \cup V(R)$ where $R$ is the $vw$-{bridge} and
    $w$ is the right neighbor of $v$ in $P_{i}$;
  \item if $v$ is the right end of $P_i$, then
    $\Scope(v) = V(L) \cup \myBox_v$, where $L$ is the $uv$-{bridge} and
    $u$ is the left neighbor of $v$ in $P_{i}$.
\end{itemize}

\noindent For an ehf-layered-wheel:
\begin{itemize}
  \item $\Scope(v) = \myBox_v$ for every $v \in P_i$, $0 \leq i \leq l-1$.
\end{itemize}

For $d \geq 0$, we also define the \emph{depth-$d$ scope} of each vertex in the layered wheel,
which will be denoted by $\Scope^d(v)$.  We
define $\Scope^0(v) =\{v\}$, and
$$ \Scope^d(v) = \bigcup_{x \in \Scope(v)} \Scope^{d-1}(x) \ \ \textrm{for } 1 \leq d \leq l-i. $$

For a layered wheel $G_{l,k}$, we define the interval graph
$I(G_{l,k})$.  For every vertex $v \in G_{l,k}$, define path $P(v)$
associated to $v$ as follows:
\begin{itemize}
\item if $v \in P_l$ is not the right end of $P_l$, then $P(v) = vw$
  where $w$ is the right neighbor of $v$;
  \item if $v$ is the right end of $P_l$, then $P(v) = \{v\}$;
  \item if $v \in P_i$ with $i < l$, then
    $P(v) = P_l \left[ \Scope^{l-i}(v) \right]$.
\end{itemize}

\noindent Note that $P(v)$ is a subpath of $P_l$. The graph
$I(G_{l,k})$ is the interval graph associated to
$\{P(v) \ | \ v \in V(G_{l,k}) \}$.

\vspace{1 ex}

\begin{lemma}
  \label{lem:subgraphI}
  For any layered wheel $G_{l,k}$ and the corresponding interval graph
  $I(G_{l,k})$, $G_{l,k}$ is a subgraph (possibly not induced) of
  $I(G_{l,k})$.
\end{lemma}

\begin{proof}
  It is clear by definition that there is a bijection between
  $V(I(G_{l,k}))$ and $V(G_{l,k})$.  We show that
  $E(G_{l,k}) \subseteq E(I(G_{l,k}))$: for any two
  vertices $u, v \in G_{l,k}$, if $uv \in E(G_{l,k})$ then the
  corresponding paths $P(u)$ and $P(v)$ share at least one vertex (i.e.\
  $V(P(u)) \cap V(P(v)) \neq \emptyset$).
  
  For $u,v \in P_l$ where $u$ is on the left of $v$, 
  this property trivially holds, because by definition,
  $P(u)$ and $P(v)$ both contain $v$. 
  If $u \in P_i$ for some $i < l$ and $v \in P_l$,
  then $V(P(v)) \subseteq V(P(u)) = \Scope^{l-i}(u)$.
  The case is similar when $v \in P_i$ for some $i < l$ and $u \in P_l$.
  
  If $u, v \in P_i$ for some $i < l$, then by definition,
  $\Scope(u) \cap \Scope(v) \neq \emptyset$ (they both contain the $uv$-bridge).  Let
  $x \in \Scope(u) \cap \Scope(v)$. Note that for $1 \leq d \leq l-i$,
  $\Scope^d(u)$ and $\Scope^d(v)$ both contain $\Scope^{d-1}(x)$.
  If $u \in P_i$ and $v \in P_j$ where
  $1 \leq i < j < l$, then $\Scope(v) \subseteq \Scope^{j-i+1}(u)$. So
  $\Scope^d(v) \subseteq \Scope^{d+j-i}(u)$ for every $1 \leq d \leq l-j$.
  The case is similar when $u \in P_j$ and $v \in P_i$ where $1 \leq i < j < l$.
  Hence, $V(P(u)) \cap V(P(v)) \neq \emptyset$.
\end{proof}


\vspace{1 ex}
\begin{theorem}
\label{th:bounded-pw}
For every integers $l \geq 2$ and $k \geq 4$, we have $\tw(G_{l, k}) \leq \pw(G_{l,k}) \leq 2l$.
\end{theorem}

\begin{proof}
  \setcounter{claim}{0} By Lemmas~\ref{lem:tw-leq-pw} (third item),
  \ref{lem:pw-interval-graph} and~\ref{lem:subgraphI}, it is enough to
  show that $\omega(I(G_{l,k})) \leq 2l+1$.  
  
  \vspace{1 ex}
  
  \begin{claim}
  \label{cl:interval-graph-clique-nb}
  Let $u$ and $v$ be non-adjacent vertices in $P_i$ for some
  $1 \leq i \leq l-1$.
  Then for any $1 \leq d \leq l-i$, we have
  $\Scope^d(u) \cap \Scope^d(v) = \emptyset$.
  \end{claim}


  \bpc 
  Let $u$ and $v$ be be non-adjacent vertices in $P_i$, where $1 \leq i \leq l-1$
  and without loss of generality, they appear in this order (from left to right) along $P_i$.
  We prove the statement by induction on~$d$.
  
  For $d = 1$, 
  it follows from the definition that $\Scope^1(u) \cap \Scope^1(v) = \emptyset$ for every possible $i$.
  Suppose for induction that $\Scope^d(u) \cap \Scope^d(v) = \emptyset$ for some $1 \leq d \leq l-i-1$.
  Note that $\Scope^d(u)$ and $\Scope^d(v)$ appear in this order along $P_{i+d}$.
  Moreover, the right end of $\Scope(u)$ and the left end of $\Scope(v)$ are also non-adjacent 
  (because they both are vertices with an ancestor).
  So for any $x \in \Scope(u)$ and $y \in \Scope(v)$, we have $xy \notin E(G_{l,k})$,
  It then follows by construction, that for every $d \geq 2$,
  for any $x \in \Scope^d(u)$ and $y \in \Scope^d(v)$, we have $xy \notin E(G_{l,k})$,
  so the induction hypothesis holds for the pair $x$ and $y$.
  We need to show that $\Scope^{d+1}(u) \cap \Scope^{d+1}(v) = \emptyset$.
  Indeed:
  $$ \Scope^{d+1}(u) \cap \Scope^{d+1}(v) = \bigcup_{x \in \Scope(u)} \Scope^{d}(x) \cap \bigcup_{y \in \Scope(v)} \Scope^{d}(y) = \emptyset, $$
  which completes our induction.
  \epc
  
  \vspace{1 ex}
  
  Let $K$ be a maximum clique in $I(G_{l,k})$.  
  By definition, for every $u,v \in P_l$ that 
  are non-adjacent, we have $V(P(u)) \cap V(P(v)) = \emptyset$.
  So no edge exists between $P_u$ and $P_v$ in $I(G_{l,k})$. 
  Similarly for non-adjacent vertices $u, v \in P_i$ where $1 \leq i \leq l-1$,
  it follows from Claim~\ref{cl:interval-graph-clique-nb}, 
  that $V(P(u)) \cap V(P(v)) = \emptyset$.  
  Therefore, $K$ contains at most two vertices of every layer $P_i$, with
  $1 \leq i \leq l$.  Since $K$ may also contain the unique vertex in
  $P_0$, then $\omega(I(G_{l,k})) \leq 2l+1$ as desired.
\end{proof}

\vspace{1ex}

The following directly follows.

\begin{corollary}
	For any integers $l \geq 2$ and $k \geq 4$, we have 
	$\tw(G_{l,k}) = \bigO \left(\log |V(G_{l,k})| \right)$.
\end{corollary}

\begin{proof}
	By Lemma~\ref{lem:neighborhood-ttf-layered-wheel} and Lemma~\ref{lem:neighborhood-ehf-layered-wheel},
	we know that $G_{l,k}$ contains at least $c \cdot 3^{l}$ vertices for some integer $c \geq 3$. Hence by Theorem~\ref{th:bounded-pw}, we have $\tw(G_{l,k}) \leq 2l \leq c' \cdot \log|V(G_{l,k})| $ for some constant $c' >0$.
\end{proof}

\section{Acknowledgement}

Thanks to \'Edouard Bonnet, Zden\v ek Dvo\v r\'ak, Serguei Norine,
Marcin Pilipczuk, Sang-il Oum, Natacha Portier, St\'ephan Thomass\'e,
Kristina Vu\v skovi\'c, and R\'emi Watrigant for useful discussions.
We are also grateful to two anonymous referees. In particular their remarks
lead us to discover a mistake in Construction~\ref{cons:G2(l,k)} that
is now fixed.

\bibliographystyle{plain}

\begin{thebibliography}{10}

\bibitem{adlerLMRTV:rwehf}
I.~Adler, N.{-}K. Le, H.~M{\"{u}}ller, M.~Radovanovi\'c, N.~Trotignon, and
  K.~Vu\v skovi\'c.
\newblock On rank-width of even-hole-free graphs.
\newblock {\em Discrete Mathematics {\&} Theoretical Computer Science}, 19(1),
  2017.

\bibitem{DBLP:conf/wg/2000}
U.~Brandes and D.~Wagner, editors.
\newblock {\em Graph-Theoretic Concepts in Computer Science, 26th International
  Workshop, {WG} 2000, Konstanz, Germany, June 15-17, 2000, Proceedings},
  volume 1928 of {\em Lecture Notes in Computer Science}. Springer, 2000.

\bibitem{CameronCH18}
K.~Cameron, S.~Chaplick, and C.~T.~Ho{\`{a}}ng.
\newblock On the structure of (pan, even hole)-free graphs.
\newblock {\em Journal of Graph Theory}, 87(1):108--129, 2018.

\bibitem{CameronSHV18}
K.~Cameron, M.V.G.~da~Silva, S.~Huang, and K.~Vu\v skovi\'c.
\newblock Structure and algorithms for (cap, even hole)-free graphs.
\newblock {\em Discrete Mathematics}, 341(2):463--473, 2018.

\bibitem{chudnovsky.r.s.t:spgt}
M.~Chudnovsky, N.~Robertson, P.~Seymour, and R.~Thomas.
\newblock The strong perfect graph theorem.
\newblock {\em Annals of Mathematics}, 164(1):51--229, 2006.

\bibitem{chudnovsky.seymour:claw4}
M.~Chudnovsky and P.~Seymour.
\newblock Claw-free graphs. {IV}. {D}ecomposition theorem.
\newblock {\em Journal of Combinatorial Theory, Series B}, 98(5):839--938,
  2008.
  
\bibitem{chudnovsky.t.t.v:ehfpyramidf}
M.~Chudnovsky, S.~Thomass{\'e}, N.~Trotignon, and K.~Vu{\v s}kovi{\'c}.
\newblock Maximum independent sets in (pyramid, even hole)-free graphs
\newblock {\em CoRR}, abs/1912.11246, 2019.

\bibitem{chuzhoy:16}
J.~Chuzhoy.
\newblock Improved bounds for the excluded grid theorem.
\newblock {\em CoRR}, abs/1602.02629, 2016.

\bibitem{ConfortiCKV00}
M.~Conforti, G.~Cornu{\'{e}}jols, A.~Kapoor, and K.~Vu\v skovi\'c.
\newblock Triangle-free graphs that are signable without even holes.
\newblock {\em Journal of Graph Theory}, 34(3):204--220, 2000.

\bibitem{CourcelleO00}
B.~Courcelle and S.~Olariu.
\newblock Upper bounds to the clique width of graphs.
\newblock {\em Discret. Appl. Math.}, 101(1-3):77--114, 2000.

\bibitem{CorneilR05}
D.~G.~Corneil and U.~Rotics.
\newblock On the relationship between clique-width and treewidth.
\newblock {\em {SIAM} J. Comput.}, 34(4):825--847, 2005.

\bibitem{CyganFKLMPPS15}
M.~Cygan, F.~V. Fomin, L.~Kowalik, D.~Lokshtanov, D.~Marx, 
M.~Pilipczuk, M.~Pilipczuk, and S.~Saurabh.
\newblock {\em Parameterized Algorithms}.
\newblock Springer, 2015.

\bibitem{diotRaTrVu:15}
E.~Diot, M.~Radovanovi\'c, N.~Trotignon, and K.~Vu{\v s}kovi{\'c}.
\newblock On graphs that do not contain a theta nor a wheel part i: two
  subclasses.
\newblock {\em CoRR}, abs/1504.01862, 2015.

\bibitem{faenzaOrioloStauffer:clawFree}
Y.~Faenza, G.~Oriolo, and G.~Stauffer.
\newblock An algorithmic decomposition of claw-free graphs leading to an
  {$O(n^{3}$)}-algorithm for the weighted stable set problem.
\newblock In {\em SODA}, pages 630--646, 2011.

\bibitem{fialaKLP:12}
J.~Fiala, M.~Kami{\'n}ski, B.~Lidick{\'y}, and D.~Paulusma.
\newblock The $k$-in-a-path problem for claw-free graphs.
\newblock {\em Algorithmica}, 62(1--2):499--519, 2012.

\bibitem{GurskiW00}
F.~Gurski and E.~Wanke.
\newblock The tree-width of clique-width bounded graphs without
  \emph{K\({}_{\mbox{n,n}}\)}.
\newblock In Brandes and Wagner \cite{DBLP:conf/wg/2000}, pages 196--205.

\bibitem{gyarfas:perfect}
A.~Gy{\'a}rf{\'a}s.
\newblock Problems from the world surrounding perfect graphs.
\newblock {\em Zastowania Matematyki Applicationes Mathematicae}, 19:413--441,
  1987.


\bibitem{king:these}
A.~King.
\newblock {\em Claw-free graphs and two conjectures on $\omega$, $\Delta$, and
  $\chi$}.
\newblock PhD thesis, McGill University, 2009.

\bibitem{OumS06}
S.-i. Oum and P. D. Seymour.
\newblock Approximating clique-width and branch-width.
\newblock {\em J. Comb. Theory, Ser. {B}}, 96(4):514--528, 2006.

\bibitem{radovanovicV:theta}
M.~Radovanovi{\'c} and K.~Vu{\v s}kovi{\'c}.
\newblock A class of three-colorable triangle-free graphs.
\newblock {\em Journal of Graph Theory}, 72(4):430--439, 2013.

\bibitem{RobertsonS86}
N.~Robertson and P. D. Seymour.
\newblock Graph minors. {V}. {E}xcluding a planar graph.
\newblock {\em Journal of Combinatorial Theory, Series {B}}, 41(1):92--114,
  1986.

\bibitem{nicolas:perfect}
N.~Trotignon.
\newblock Perfect graphs: a survey.
\newblock {\em CoRR}, abs/1301.5149, 2013.

\bibitem{truemper}
K.~Truemper.
\newblock Alpha-balanced graphs and matrices and {GF}(3)-representability of
  matroids.
\newblock {\em Journal of Combinatorial Theory, Series B}, 32:112--139, 1982.

\bibitem{vuskovic:evensurvey}
K.~Vu{\v s}kovi{\'c}.
\newblock Even-hole-free graphs: a survey.
\newblock {\em Applicable Analysis and Discrete Mathematics}, 10(2):219--240,
  2010.

\bibitem{vuskovic:truemper}
K.~Vu{\v s}kovi{\'c}.
\newblock The world of hereditary graph classes viewed through {T}ruemper
  configurations.
\newblock In S.~Gerke S.R.~Blackburn and M.~Wildon, editors, {\em Surveys in
  Combinatorics, London Mathematical Society Lecture Note Series}, volume 409,
  pages 265--325. Cambridge University Press, 2013.

\bibitem{watkinsMesner:cycle}
M.~E.~Watkins and D.~M.~Mesner.
\newblock Cycles and connectivity in graphs.
\newblock {\em Canadian Journal of Mathematics}, 19:1319--1328, 1967.

\end{thebibliography}

\end{document}